\documentclass[a4paper,twocolumn,11pt]
{quantumarticle}
\pdfoutput=1
\usepackage{color}
\usepackage{bbm}
\usepackage{graphicx}
\usepackage{dcolumn}
\usepackage[utf8]{inputenc}

\usepackage{graphicx}
\usepackage{epstopdf}
\usepackage{amsthm}
\usepackage{amsmath}
\usepackage{empheq}
\usepackage{bbm}
\usepackage{braket}
\usepackage{amssymb}
\usepackage[usenames,dvipsnames]{xcolor}
\usepackage{cases}
\usepackage{latexsym}
\usepackage[colorlinks=true,citecolor=Cerulean,linkcolor=RubineRed,urlcolor=Cerulean]{hyperref}
\usepackage[title]{appendix}
\usepackage{physics}

\usepackage[normalem]{ulem}

\usepackage[numbers,compress]{natbib}

\graphicspath{ {images/}{images/figure1/}{images/figure2/}{images/figure3/}{images/figure4/}{images/figure5/} }

\newtheorem{theorem}{Theorem}

\newtheorem{lemma}{Lemma}
\newtheorem{corollary}{Corollary}

\newtheorem{definition}{Definition}

\newtheorem{assumption}{Assumption}

\newcommand{\alv}[1]{{#1}}
\newcommand{\ksr}[1]{{ #1}}

\newcommand{\newedit}[1]{{#1}}
 
\usepackage{lipsum}
\usepackage{graphicx}

\begin{document}
\raggedbottom	
\title{Matrix product state approximations to quantum states of low energy variance }
	
\author{Kshiti Sneh Rai}
\email{ksrai@lorentz.leidenuniv.nl}
\affiliation{Instituut-Lorentz, Niels Bohrweg 2, Leiden, NL-2333 CA, The Netherlands}
\affiliation{Max-Planck-Institut für Quantenoptik, Hans-Kopfermann-Straße 1, D-85748 Garching, Germany}

\author{J. Ignacio Cirac}
\affiliation{Max-Planck-Institut für Quantenoptik, Hans-Kopfermann-Straße 1, D-85748 Garching, Germany}

\author{\'Alvaro M. Alhambra}
\email{alvaro.alhambra@csic.es}
\affiliation{
Instituto de F\'isica Te\'orica UAM/CSIC, C/ Nicol\'as Cabrera 13-15, Cantoblanco, 28049 Madrid, Spain}
\affiliation{Max-Planck-Institut für Quantenoptik, Hans-Kopfermann-Straße 1, D-85748 Garching, Germany}

\begin{abstract}

We show how to efficiently simulate pure quantum states in one dimensional systems that have both finite energy density and vanishingly small energy fluctuations. We do so by studying the performance of a tensor network algorithm that produces matrix product states whose energy variance decreases as the bond dimension increases. Our results imply that variances as small as $\propto 1/\log N$ can be achieved with polynomial bond dimension. With this, we prove that there exist states with a very narrow support in the bulk of the spectrum that still have moderate entanglement entropy, in contrast with typical eigenstates that display a volume law. Our main technical tool is the Berry-Esseen theorem for spin systems, a strengthening of the central limit theorem for the energy distribution of product states. We also give a simpler proof of that theorem, together with slight improvements in the error scaling, which should be of independent interest. 
	\end{abstract}
	
	\maketitle

\section{Introduction}

It is widely established that entanglement is one of the most important concepts in the study of quantum many-body systems. The main reason why is that the character of entanglement in a system can typically be connected to fundamental physical properties. For instance, an area law in the low-energy states is associated with absence of criticality, localized correlations, and tensor network approximations \cite{Verstraete_2006,Eisert_2010}. On the other hand, a larger amount of entanglement in the ground state can be associated to the appearance of quantum phase transitions \cite{CriticalEntanglement}. 

The entanglement properties in the bulk of the spectrum, beyond the low energy sector, are also of crucial importance. For instance, the energy eigenstates of finite energy density (with zero energy variance) most often have large amounts of entanglement, compatible with their local marginals resembling Gibbs states as per the Eigenstate Thermalization Hypothesis (ETH) \cite{Srednicki_1994,DeutschETH1991}. In contrast, product states are also in the bulk of the spectrum, but their lack of entanglement comes hand in hand with a larger energy variance.

Eigenstates and product states are the two extreme situations of either no energy variance or no entanglement. However, is this a fundamental trade-off? Or, alternatively, are there states that have both low entanglement and small energy variance?
Here, we study this intermediate regime by answering the following question: what is the entanglement generated when narrowing down the variance of an initial quantum state?
We do this by rigorously analyzing the performance of a matrix product state algorithm inspired by that in \cite{Banuls_2020} that decreases the energy variance of any initial product state, while only increasing the bond dimension in a controlled manner.
\newedit{In \cite{Banuls_2020}, a heuristic analysis of the algorithm was given, with a similar expression for the performance which suggested that small variances were possible in arbitrary Hamiltonians. 
They also numerically showed that for specific ones obeying the ETH, arbitrarily small variances for the final state can be achieved. 
This analysis, however, did not consider the most general cases. 
In fact, we show that some simple counterexamples are not able to reach arbitrarily small variances, and that these are related to cases in which the distance of the wavefunction from a Gaussian is maximal. Our rigorous results explicitly show how in general the larger the deviation, the more limited one is in achieving smaller variances.}

As a main result, we show that vanishingly small energy variances $\delta^2$, down to $\delta^2 \propto 1/\log(N)$, can be achieved with polynomial bond dimension, or $\log N$ entanglement entropy, in many models of interest. 
\newedit{We demonstrate the accuracy of our results by discussing particular cases in which the bounds are tight (Sec.~\ref{sec:product}), corresponding to examples where the support of the initial product state limits how much the variance can be decreased. These models also include the counterexamples to the heuristic results in \cite{Banuls_2020}.}

Our method for studying the energy distribution at the output of the algorithm is based on the Berry-Esseen theorem \cite{berry1941accuracy,esseen1942liapounoff} for spin systems. This result quantifies how much the energy distribution of product states resembles a Gaussian. We give a short proof of it based on the cluster expansion result from \cite{Wild_2022}, improving previous ones by a polylogarithmic factor \cite{Berry-Esseen}. This renders the bound optimal, and should be of independent interest. 
In our proofs, we use that approximate Gaussianity to estimate the energy average and variance of the  state after a \emph{filtering operator} has been applied to it. This filter is designed to narrow down the energy fluctuations, while being expressible as a tensor network with a bond dimension that can be easily upper bounded.

An additional motivation for this scheme is that it is expected that, under the ETH, any state with a low enough energy variance will also locally resemble a Gibbs state \cite{Dymarsky2019,Banuls_2020}. This idea is the basis of heuristic quantum algorithms for measuring local thermal expectation values \cite{Lu2021,Schuckert_2023}, in a way that is potentially amenable to near term quantum simulators \cite{Daley2022}. With our bounds on the bond dimension of the MPS scheme, we rigorously analyze the classical simulability of those algorithms. In that sense, we narrow down the situations in which such a scheme will be computing quantities beyond the reach of classical computers. In doing so, we also offer efficiency guarantees for related tensor network algorithms \cite{_akan_2021,Lu2021,Yilun2022}.

The article is arranged as follows. We introduce the setting in Sec. \ref{sec:prelims}, and analyze the energy distribution of product states in Sec. \ref{sec:product}. The energy filter is introduced in Sec. \ref{sec:filter}, and in Sec. \ref{sec:algo} we upper bound its bond dimension. We specialize to ETH Hamiltonians in Sec. \ref{sec:ETH} and conclude. The more technical proofs are placed in the appendices.


\section{Preliminaries}\label{sec:prelims}

We focus on systems of $N$ particles described by a local Hamiltonian
\begin{equation}
    H=\sum_{i=1}^N h_i=\sum_j E_j \ketbra{E_j}{E_j},
\end{equation}
where each of the Hamiltonian terms is such that $\norm{h_i}\le 1$, and overlaps with a small $\mathcal{O}(1)$ number of other qubits. If the Hamiltonian is in 1D (which we assume in Sec. \ref{sec:algo}), this refers to adjacent sites on the chain. The spectrum $\{E_j \}$ is thus confined to the range $[-N,N]$. The initial states that we consider are product among all particles
\begin{equation}
    \ket{p}=\bigotimes_i \ket{p_i}=\sum_{E_j} b_j \ket{E_j},
\end{equation}
so that $b_j$ are the coefficients of the wavefunction in the energy eigenbasis. Their average and variance are
\begin{align}
    &E= \bra{p} H \ket{p} = \sum_{E_j} \vert b_j \vert^2 E_j, \\ &\sigma^2 = \bra{p} (H-E)^2 \ket{p} = \sum_{E_j} \vert b_j \vert^2 (E_j-E)^2.
\end{align}
A mild but important assumption throughout this work is the following. 
\begin{assumption}\label{ass1}
The ratio
 $\frac{\sigma}  {\sqrt{N}} \equiv s $ is $\Omega(1)$ for all $N$.
\end{assumption}
\alv{The symbol $\Omega(1)$ means that the quantity is at least as large as a constant factor, independent of $N$.} Also, it can be easily shown that, for product states, $\sigma = \mathcal{O}\left (\sqrt{N} \right)$. Thus, we are only assuming that the variance scales as fast as possibly allowed, which is a generic property of product states. With this, we rule out trivial situations in which e.g. $\ket{p}$ is already an eigenstate of the Hamiltonian \alv{or fine-tuned situations in which the energy variance of a large region in the system is made artificially small}. Our goal will be to reduce the energy variance of these product states by applying a filtering operator, while controlling the amount of entanglement generated in the filtering process.

\section{Gaussian energy distribution} \label{sec:product}

Due to the lack of correlations among all the sites, the energy distribution of a product state shares features with the distribution of $N$ independent random variables. Along these lines, the central limit theorem \cite{Hartmann_2004} and the Chernoff-H\"offding bound \cite{anshu2016concentration,Kuwahara_2016} have previously been shown.

Here, we focus on a closely related but stronger result: 
the Berry-Esseen theorem \cite{berry1941accuracy,esseen1942liapounoff}. Originally, this showed that the cumulative distribution function of $N$ random variables is close to that of a Gaussian with the same average and variance, up to an error $\mathcal{O}(N^{-1/2})$. To introduce it in our context, let us define the following.

\begin{definition}\label{def:Improved_BE}   
Consider the cumulative function
    \begin{equation}
        J(x) := \sum\limits_{E_j\leq x}{\vert b_j\vert^2},
    \end{equation}
    and the corresponding Gaussian cumulative function
    \begin{equation}
        G(x) := \int\limits_{-\infty}^{x}{\frac{dt}{\sqrt{2\pi\sigma^2}}e^{-\frac{(t-E)^2}{2\sigma^2}}}.
    \end{equation}
    The Berry-Esseen error $\zeta_N$ is 
    \begin{equation}
        \sup_x{\vert J(x)-G(x)\vert} \equiv \zeta_N.
        \label{eq:berry-esseen}
    \end{equation}
\end{definition}
The generalization of this definition to arbitrary states, and in particular mixed ones, is straightforward. The scaling of the error $\zeta_N$ with $N$ quantifies how much does the energy distribution deviate from the normal distribution. Our first technical result is the following.

\begin{lemma}\label{le:BEerror}
Let $\ket{p}$ be a product state obeying Assumption \ref{ass1}. Then,
\begin{equation}\label{eq:BEsca}
    \zeta_N = \mathcal{O}(N^{-1/2}).
\end{equation}
\end{lemma}
The proof is shown in Appendix \ref{sec:Improved_BE_proof}. It follows straightforwardly from the original Esseen's inequality, together with the cluster expansion results from \cite{Wild_2022}. Notice that this Lemma holds for all Hamiltonians that are few-body local, even beyond 1D. 

Lemma \ref{le:BEerror} can be seen as a qualitative strengthening of \cite{Hartmann_2004}, where it was shown that in 1D $\lim_{N\rightarrow \infty} \zeta_N=0$.
It also gives a poly-logarithmic improvement on the best previous bound \cite{Berry-Esseen,Brandao_2015G} which showed that $\zeta_N = \mathcal{O}\left(\frac{\log^{2D}N}{\sqrt{N}}\right)$ in a $D$ dimensional lattice. That result, however, also holds for the more general class of states with exponential decay of correlations, for which we expect the logarithmic factor is necessary \cite{TikhomirovBE81}.

The scaling of Eq. \eqref{eq:BEsca} is tight up to constant factors. For a specific example matching the bound, consider $N$ independent random coin tosses with equal probability, so that the probability of $k$ tails is (see Fig. \ref{fig:binomial})
\begin{equation}\label{eq:binom}
    p_k= \frac{1}{2^N}\binom{N}{k}.
\end{equation}
Let $N$ be odd. Then, for $\frac{1}{2}N \ge x \ge \frac{1}{2}(N-1)$,
\begin{align}\label{eq:xline}
\vert J(x)-G(x) \vert &= \left \vert \frac{1}{2}-\int\limits_{-\infty}^{x}{\frac{dt}{\sqrt{2\pi\sigma^2}}e^{-\frac{(t-\frac{1}{2}N)^2}{2\sigma^2}}} \right \vert 
\\& =\left \vert \int\limits_{0}^{\frac{1}{2}N-x}{\frac{dt}{\sqrt{2\pi\sigma^2}}e^{-\frac{(t-\frac{1}{2}N)^2}{2\sigma^2}}} \right \vert \nonumber 
\\ & \ge \frac{1}{\sqrt{\pi}}\left(\frac{\sqrt{2}\left(\frac{1}{2}N-x \right)}{\sigma}\right.\nonumber
\\&\quad\quad\quad\quad\quad\quad\left.
-\frac{\left(\frac{1}{2}N-x \right)^3}{ \sqrt{2}\sigma^3} \right).\nonumber
\end{align}
Choosing e.g. $x=\frac{1}{2}(N-\frac{1}{2})$ means the lower bound is $\Omega(\sigma^{-1})=\Omega(N^{-1/2})$. 
\begin{figure}[t!]
\centering
\includegraphics[trim=0.02cm 0.2cm 0 0, clip,scale=0.68]{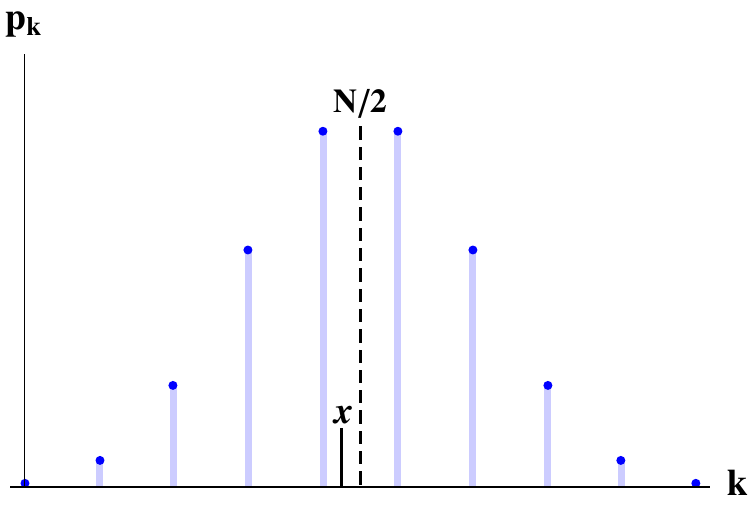}
\caption{Representation of the binomial distribution $p_k$ in Eq. \eqref{eq:binom} for an odd value of $N$. For a non-interacting Hamiltonian with this distribution, the average energy is at $N/2$, where there is no eigenstate. The smallest standard deviation one can reach by filtering this distribution is limited to the difference between eigenergies at $k=\frac{N-1}{2}$ and $k=\frac{N+1}{2}$. We also illustrate a possible choice of $x$ in the derivation of Eq. \eqref{eq:xline}.} 
\label{fig:binomial}
\end{figure}

Importantly, there are product states on spin systems with exactly Eq. \eqref{eq:binom} as the energy distribution. 
A trivial example is the state $\bigotimes_i^N \ket{+}_i$ with a non-interacting Hamiltonian $H=\sum_i \sigma_{Z,i}$. 
\newedit{Additionally}, it also appears in certain Hamiltonians with strong interactions, such as the model considered in \cite{Schecter_2019}, in which the product state has support on a $\mathcal{O}(N)$ number of eigenstates called \emph{quantum scars} \cite{Turner_2018,AlvaroScars}.
\newedit{These models are counterexamples to the generic case assumed in \cite{Banuls_2020}.
In Sec.~\ref{sec:ETH}, we restrict to ETH Hamiltonians, and explain how in those generic cases these counterexamples are not present.}






\section{Energy filter}\label{sec:filter}

We now define an operator which, when applied to any quantum state (in this case, our initial product state), will decrease its energy variance.
\begin{figure}[t!]
\centering
\includegraphics[trim=2.3cm 0.085cm 1.9cm 1.1cm, clip,scale=0.66]{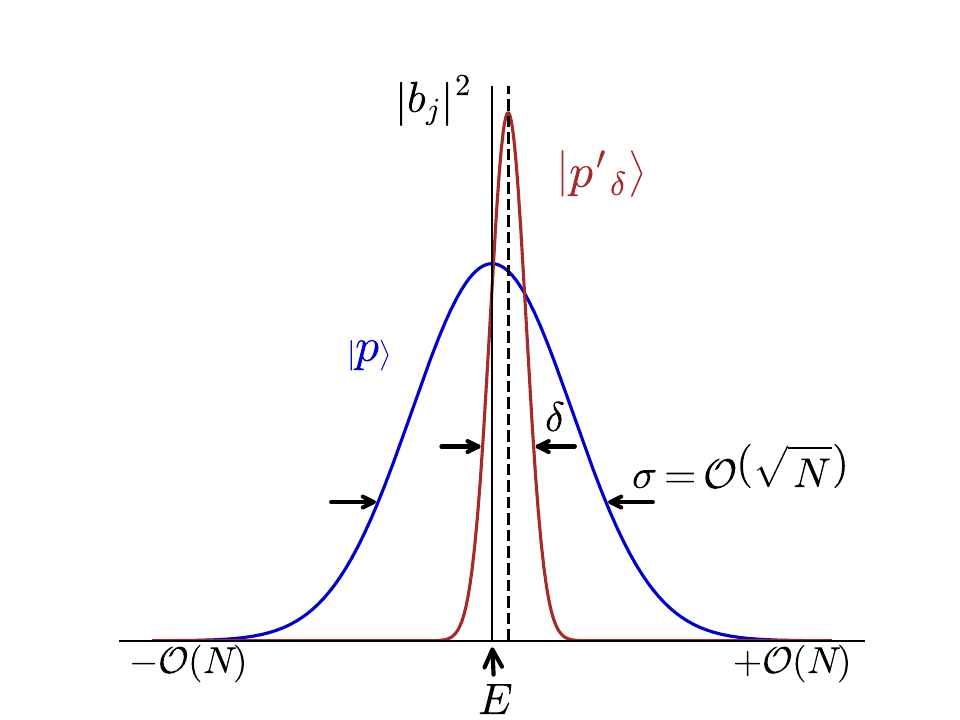}
\caption{ Illustration of the energy distribution of the initial product state, with standard deviation $\sigma$, and of the filtered distribution, with a smaller standard deviation $\delta$ and potentially a slightly shifted average (dashed line).}
\label{fig:filter}
\end{figure}
To do this, we first consider the \emph{cosine filter } \cite{ge2018faster,Banuls_2020}
\begin{equation}
    P_M(E) = \cos^M\left(\frac{H-E}{N}\right).
    \label{eq:cosine_filter}
\end{equation}
For an illustration of the effect of the filter on the energy distribution of the product state, see Fig. \ref{fig:filter}.

Since the eigenvalues of $\frac{H-E}{N}$ are within $[-1,1]$, the cosine function acts as an approximate projector around the energies close to $E$ when $M$ is very large. An advantage of this operator is that using the binomial expansion we can write
\begin{equation}\label{eq:cosine-filter2}
    \cos^M\left(\frac{H-E}{N}\right)= \frac{1}{2^{M}}\sum\limits_{m=-\frac{M}{2}}^{\frac{M}{2}}c_{m}e^{i2m(H-E)/N}.
\end{equation}
where $c_m = \binom{M}{\frac{M}{2}-m}$. There is some freedom as to how to define the denominator inside the cosine filter \cite{Yilun2022}, which does not change the final result significantly, so for simplicity we choose $N$ as in Eq. \eqref{eq:cosine_filter}. 

Eq. \eqref{eq:cosine-filter2} shows that when the full operator $P_M(E)$ is applied to a trial state, the result can be written as a superposition of time evolution operators. Additionally, the amount of complex exponentials can be significantly reduced for a small cost in precision. 

\begin{lemma}\label{def:grouped_cosine}
    Let the approximate cosine function $g_y(x)$ be
    \begin{equation}
        g_y(x) := \frac{1}{2^{M}}\sum\limits_{m=-y\sqrt{M}}^{y\sqrt{M}}{\binom{M}{\frac{M}{2}-m}e^{i2mx}}.
    \label{eq:gyx-def}
    \end{equation}
    This is such that, $\forall x \in [-1,1]$,
\begin{equation}
 \vert  \cos^M(x) - g_y(x) \vert \leq e^{-\frac{y^2}{2}}.
\end{equation}
\end{lemma}
This is obtained from the Chernoff bound on the binomial coefficients, which is exponentially decaying in $y^2$. For a significant reduction in the number of terms, we need $y$ to be $o(\sqrt{M})$. 

Lemma \ref{def:grouped_cosine} allows us to define our filtered state as
\begin{equation}
    \ket{p_{M,y}}= \frac{g_y\left(\frac{H-E}{N}\right) \ket{p}}{\norm{g_y\left(\frac{H-E}{N}\right) \ket{p}}},
    \end{equation}
which leads to the main technical result of this section.

\begin{lemma}\label{le:Variance_bound_general}
Let $M=\mathcal{O}\left(\frac{N}{\zeta_N^2}  \right)$ and $M =  \Omega(N)$, and $y= \Omega\left(\sqrt{\log \left(M^3/N^2 \right) } \right)$. The energy average $\mu$ and variance $\delta^2$ of $\ket{p_{M,y}}$ are bounded as
    \begin{align} \label{eq:avlemma}
    \vert \mu - E \vert & = 
    \mathcal{O}\left(\frac{N}{\sqrt{M}}\right),
    \end{align}
    \begin{align}\label{eq:varlemma}
        \delta^2 & = \mathcal{O}\left(\frac{N^2}{M}\right).
    \end{align}
\end{lemma}
This lemma only applies for large enough values of $M$. On the other hand, for $M=\mathcal{O}(N)$, the best bounds one can achieve are
\begin{align}
    \vert \mu - E \vert & =
    \mathcal{O}\left(\sqrt{N}\right), \\
        \delta^2 & = \mathcal{O}\left(N\right),
\end{align}
which means that the filter does not change the variance in a substantial way - at most up to a constant factor.

For large $M$, the lemma shows that the filter $g_y\left(\frac{H-E}{N}\right)$ decreases the energy variance of the state without changing the average energy too much. The proof is shown in App. \ref{app:proofMain}. It is based on the Berry-Esseen error from Definition \ref{le:BEerror}, which we use to approximate the expression of the variance $\delta^2$ as a Gaussian integral, with a considerable amount of error terms that need to be estimated. \alv{The overall argument is a way of segmenting energy sums in pieces, so that via the approximation of Lemma \ref{le:BEerror} they can be turned into pieces of a tractable Gaussian integral. A similar argument previously appeared in a different context in \cite{Hovhannisyan_2021}.} We note that this holds for Hamiltonians that are local, independent of the dimension. The upper bound on $M$ in terms of $\zeta_N^{-1}$ is due to the fact that the deviation from Gaussian limits how much we can resolve the energy distribution after the filter is applied: if the deviation is large, the filter may act in a more uncontrolled way, and the bound may not hold.  

In many practical scenarios we expect the wavefunction $\ket{p}$ to be symmetric about the energy average, in which case $\mu$ and $E$ will be identical. On the other hand, the scaling of the bound on $\delta$ is tight up to constant factors: if the wave-function coefficients $b_j$ are taken to be an exact Gaussian, one can explicitly calculate $\delta \simeq \frac{N}{\sqrt{2M}}$ (see Eq. (15) in \cite{Banuls_2020}).

The result is purposely stated in terms of the general error from Def. \ref{def:Improved_BE}. Lemma \ref{le:BEerror}, however, allows us to establish that
choosing $ M \propto N^2$, the bound on the variance is
\begin{equation}
  \delta^2 = \mathcal{O}(1).
\end{equation}

This means that for any initial product state with standard deviation of $\mathcal{O}(\sqrt{N})$, the cosine filter can be applied to reduce the standard deviation down to $\mathcal{O}(1)$, and shifting the average by at most that amount. 
\ksr{Note that a similar variance bound (Eq.~\eqref{eq:varlemma}) was estimated in \cite{Banuls_2020}. 
In this section, we have provided a rigorous proof for it, and showed that it is only valid for particular range of values of $M$ in the cosine filter, which depend on $\zeta_N$, quantity appearing from the Berry-Esseen theorem.}
\alv{Intuitively, this is because, in order to estimate the variance $\delta$, we need a high resolution of the energy populations $\vert b_j \vert^2$, which is only possible if $\zeta_N$ is small enough.}
\ksr{Larger values for $M$ are responsible for filtering a narrower range of energies by implementing an approximate eigenstate projector around energy $E$. As a result, the norm of the final state becomes exceedingly smaller with increasing $M$. However, to retain the Gaussian approximation from the Berry-Esseen theorem (Lemma~\eqref{le:BEerror}), we require the norm of the filtered state to be above a certain constant ($=\mathcal{O}(1)$), which forces an upper bound on $M$.}
The examples from Sec. \ref{sec:product} show that this is the smallest variance one can reach in general (see Fig. \ref{fig:binomial}). However, if the factor $\zeta_N$ would decrease more quickly than $\mathcal{O}(N^{-1/2})$, the range of values which $M$ can take becomes larger, and lower variances can be achieved, as we discuss in Sec. \ref{sec:ETH}. 


\section{Low variance state using MPS}\label{sec:algo}

In this section we prove our main result on the approximability of the \alv{low variance} filtered states with matrix product states. \ksr{To do this, we study the performance of the cosine filter algorithm for preparing arbitrarily low energy variance states in the finite energy density regime of a 1D local Hamiltonian. We reiterate that producing these states is interesting for the following reasons: $i)$ understanding the fundamental tradeoffs between entanglement and energy variance in the bulk of the spectrum, $ii)$ producing the simplest possible states that can approximate thermal properties locally,
and $iii)$ preparing initial states for filtering quantum algorithms, designed for constructing low-variance states inaccessible using tensor network techniques.
} 


The main idea is that
it is possible to construct a matrix product operator $g_y^D$ that approximates $g_y\left(\frac{H-E}{N}\right)$ in one spatial dimension. When applied to our initial product state $\vert p \rangle$, this implies that one can efficiently find an MPS approximation to $\vert p_{M,y} \rangle$ with a close enough average and variance in energy, while also having a small bond dimension. This is contained in our main result Theorem \ref{th:main}.

The starting point is the following lemma from \cite{Kuwahara_2021}, which shows there is a Matrix Product Operator (MPO) \cite{Pirvu_2010} approximation to the matrix exponential of $H$ .

\begin{lemma}\label{le:MPOt}\cite{Kuwahara_2021} Let $H$ be a 1D local Hamiltonian. There is an algorithm that outputs an MPO approximation $T_{t}$ of $e^{-i t H}$ such that
\begin{equation}\label{eq:time_evolution_approximation}
\norm{ T_{t} - e^{-it H}} \le \epsilon,
\end{equation}
with bond dimension upper bounded by
\begin{equation}\label{eq:MPO_bond_dimension}
D \le e^{\mathcal{O}(t) + \mathcal{O}\left(\sqrt{t \log(N/\epsilon)} \right) }.
\end{equation}
\end{lemma}

The idea here is that $g_{y}$ can be written as a sum of complex exponentials as per Eq. \eqref{eq:cosine-filter2}. Thus, we can approximate each term in the filter with Lemma \ref{le:MPOt}, and add them to obtain our approximate filter 
\begin{equation}
g_y^{D}= \frac{1}{2^{M}}\sum\limits_{m=-y\sqrt{M}}^{m=+y\sqrt{M}}{\binom{M}{\frac{M}{2}-m}T_{2m/N}}.
\end{equation}
When applied to the initial product state, this yields the final MPS with bounded variance and bond dimension. The precise result is as follows.

\begin{theorem}\label{th:main}
Let $H$ be one-dimensional, and let $\ket{p}$ satisfy Assumptions \ref{ass1}. The MPS $\vert p'_{\delta} \rangle=\frac{g_y^D\ket{p}}{\vert \vert g_y^D\ket{p} \vert \vert}$, under the conditions of Lemma \ref{le:Variance_bound_general}, has the properties
\begin{align} \label{eq:av1}
 \vert \langle p'_{\delta} \vert H  \vert p'_{\delta} \rangle - E \vert & =  \mathcal{O}\left(\frac{N}{\sqrt{M}} \right) \\ \label{eq:var1}
  \langle p'_{\delta} \vert H^2  \vert p'_{\delta} \rangle- \left(\langle p'_{\delta} \vert H  \vert p'_{\delta} \rangle\right)^2 &\le  2 \delta^2 = \mathcal{O}\left(\frac{N^2}{M} \right).
\end{align}
Moreover, the bond dimension of $g_y^D$ (and thus of $\ket{p'_\delta}$) is bounded by
    \begin{equation}\label{eq:BDBound}
        D \le \frac{N\log N}{\delta}e^{\mathcal{O}\left(\frac{\sqrt{\log (N/\delta)}}{\delta}\right ) + \mathcal{O}\left(\frac{\log^{\frac{3}{4}} (N /\delta) } {\sqrt{\delta}}\right)}.
    \end{equation}
\end{theorem}

\begin{proof}


First we have that, by Lemma \ref{le:MPOt} and the triangle inequality,
\begin{align}\label{eq:filter_approx_error}
\begin{split}
    \norm{g_y \left(\frac{H-E}{N} \right) - g_y^{D}}&\leq 2^{-M}\sum_{\vert m\vert\leq y\sqrt{M}}\binom{M}{\frac{M}{2}-m}\epsilon
    \\
    &\leq\epsilon,
\end{split}
\end{align}
where we consider the approximation error $\epsilon$ coming from Eq. \eqref{eq:time_evolution_approximation}. Thus, for the states
\begin{align}
 \vert \vert \ket{p_{M,y}}   - \ket{p'_{\delta}} \vert \vert &= \vert \vert \frac{g_y \left(\frac{H-E}{N} \right)\ket{p}}{\norm{g_y \left(\frac{H-E}{N} \right)\ket{p}}}  - \frac{g_y^{D}\ket{p}}{\norm{g_y^{D}\ket{p}}} \vert \vert 
 \\ & \le 2 \frac{\norm{g_y \left( \frac{H-E}{N} \right) - g_y^{D}} }{\norm{g_y \left(\frac{H-E}{N}  \right)\ket{p}}} \\
 & \le   \frac{2 \ \epsilon}{\norm{g_y \left(\frac{H-E}{N}  \right)\ket{p}}}.
\end{align}
The factor $\norm{g_y \left(\frac{H-E}{N}  \right)\ket{p}}$ is lower bounded in Eq. \eqref{eq:M21} as, 


\begin{align}
    \norm{g_y\left(\frac{H-E}{N}\right) \ket{p}}=\Omega\left(\frac{N^{1/2}}{M^{1/4}}\right).
\end{align}
Now, notice that  for $\ket{p_{M,y}}$ and $\ket{p'_{\delta}}$ to have the same average and variance up to the errors of Eq. \eqref{eq:av1} and \eqref{eq:var1} it is enough that,
\begin{equation} \label{eq:normbound}
    \vert \vert \ket{p_{M,y}}   - \ket{p'_{\delta}} \vert \vert \le  \frac{\delta^2}{2 N^2}.
\end{equation}
 Given Eq. \eqref{eq:normbound}, the error we require from Lemma \ref{le:MPOt} is $\epsilon=\mathcal{O}\left(\frac{\delta^2}{N^{3/2}M^{1/4}}\right)$. With this, notice that the leading error for the average is given by Eq. \eqref{eq:avlemma}.

To estimate the total bond dimension of $g_y^D$, notice that we have $2y\sqrt{M}$ terms, and that, considering Lemma \ref{le:Variance_bound_general}, the longest $t$ to consider is $y\sqrt{M}/N =  \mathcal{O}(y \delta^{-1})$. Using Lemma \ref{le:MPOt} and simplifying the powers in the logarithm, we can upper bound
\begin{align}\label{eq:final_bond_dimension}
    D \le \frac{2y N}{\delta} e^{\mathcal{O}\left(y\delta^{-1}\right ) + \mathcal{O}\left(\sqrt{y \delta^{-1} \log \left(\frac{N}{\delta} \right)}\right)}
\end{align}

Finally, choosing $y^2=6\log \frac{N}{\delta}$ as allowed by Lemma \ref{le:Variance_bound_general}, we obtain
\begin{align}
\label{eq:bond_dimension_N_delta}
    D\le \frac{N\log N}{\delta}e^{\mathcal{O}\left(\frac{\sqrt{\log (N/\delta)}}{\delta}\right ) + \mathcal{O}\left(\frac{\log^{\frac{3}{4}} (N /\delta) } {\sqrt{\delta}}\right)}.
\end{align}

\end{proof}

Additionally, the run-time of the algorithm that constructs $\vert p'_{\delta} \rangle$ is simply the bond dimension multiplied by a factor of $N$, to account for the cost of manipulating the tensors of the MPO \cite{Kuwahara_2021}. \alv{A similar estimate of the bond dimension was given in \cite{Banuls_2020}, but their argument omitted some key steps in the analysis, such as the effect of using the Chernoff bound, and the explicit distance of the wavefunction to a Gaussian, which appears through Lemma \ref{le:Variance_bound_general}.}

The most important particular case of this result is if we aim for a variance $\delta= \mathcal{O}(1)$, in which case we have the following.
\begin{corollary}\label{co:constant_variance}
   There exists an MPS $\ket{p'}$ of quasilinear bond dimension 
$ D \le N e^{\mathcal{O}(\log^{\frac{3}{4}} (N))}$
   such that
   \begin{align}
 \vert \langle p' \vert H  \vert p' \rangle - E \vert &=  \mathcal{O}(1) \\
  \langle p' \vert H^2  \vert p' \rangle- (\langle p' \vert H  \vert p' \rangle)^2 &=  \mathcal{O}(1).
\end{align}
\end{corollary}
This is achieved by choosing $M \propto N^2$ in the filter, given Theorem \ref{th:main} together with Lemma \ref{le:Variance_bound_general}. This is the smallest variance that can be achieved in general, as illustrated by the examples with the binomial energy distribution in Eq. \ref{eq:binom}. 

If instead of a product state our initial state has exponential decay of correlations, the corresponding standard deviation one can prove to reach in general is $\delta=\mathcal{O}(\log^{2D}N)$ as per the result in \cite{Brandao_2015G}. Additionally, if the Berry-Esseen error $\zeta_N$ is much smaller than the upper bound of Lemma \ref{le:BEerror} we can also reach smaller variances, down to $\zeta_N^2 \times N$, at an additional cost in the bond dimension. In particular, we expect \alv{that in practice} it is most often possible to reach $\delta^2\propto 1/\log{N}$ while still keeping the bond dimension from Eq. \eqref{eq:BDBound} polynomial in $N$.




Theorem \ref{th:main} is restricted to one dimension. For higher dimensions, the existence of a similar PEPO is guaranteed by the results of \cite{Molnar_2015}, which in our setting yields a bond dimension
\begin{align}\label{eq:bdPEPO}
D =\left( \frac{N}{\delta} \right)^{O (\frac{\sqrt{\log N}}{\delta})}.
\end{align}
This, however, is only polynomial for larger variances $\delta^2 = \Omega(\log N)$ and in no case guarantees a provably efficient approximation algorithm, considering the difficulty of contracting PEPS \cite{Schuch_2007}. 

The upper bounds on the bond dimension also guarantee bounds on the entanglement entropy of the state $\ket{p'_\delta}$. Defining $S(\rho'_A)=-\tr[\rho'_A \log \rho'_A]$ with $\rho'_A$ the marginal on region $A$, we have that $S(\rho'_A) \le \vert \partial_A \vert \log D$, with $\vert \partial_A \vert$ the number of bonds between region $A$ and its complement. In fact, this upper bound applies to all R\'enyi entanglement entropies $S_\alpha(\rho'_A)=\frac{1}{1-\alpha} \log\tr[\rho'^\alpha_A ]$ with $\alpha>0$. 

Also notice that, by the Alicki-Fannes inequality together with Eq. \eqref{eq:normbound}, the upper bound on the entanglement entropy also holds for the state $\ket{p_{M,y}}$, in which we have applied the exact filter $g_y \left(\frac{H-E}{N} \right)$. Overall, we can conclude that in 1D, both $\ket{p'_\delta}$ and $\ket{p_{M,y}}$ (with its corresponding marginal $\rho_A$) have an entanglement entropy on a region $A$  bounded by
\begin{align}
  S(\rho'_A), &S(\rho_A) \le  \mathcal{O}\left(\log (N/\delta)\right )+
  \\& 
  +\mathcal{O}\left(\frac{\sqrt{\log (N/\delta)}}{\delta}\right )+\mathcal{O}\left(\frac{\log^{\frac{3}{4}} (N /\delta) } {\sqrt{\delta}}\right). \nonumber
\end{align}
For instance, for $\delta = \Omega(1)$ we obtain the bound $S \le \mathcal{O}(\log N)$,
which is significantly smaller than the largest possible volume law scaling $S(\rho_A) \propto N$.

\section{ETH Hamiltonians}\label{sec:ETH}

So far, we have expressed many of our results, and in particular the range of $\delta$ achievable in Lemma \ref{le:Variance_bound_general}, without lower bounding the value of $\zeta_N$. Given the examples of Sec. \ref{sec:product}, the error $\mathcal{O}(N^{-1/2})$ in Lemma \ref{le:BEerror} is the strongest general upper bound on $\zeta_N$. However, in most systems of interest much more favourable scalings should hold, as we now illustrate. 

Let $S_2(E_j)=-\log \tr{\rho_{A,j}^2}$ be the R\'enyi-$2$ entanglement entropy of eigenstate $\ket{E_j}$ for an arbitrary bipartition into subsystems $A,B$ with $\vert A \vert ,\vert B \vert\propto N$. It was shown in \cite{Wilming_2019} that, given that $\ket{p}$ is a product state,
\begin{equation}\label{eq:state_populations}
 \vert b_j \vert^2 \le e^{-\frac{S_2(E_j)}{2}}.
\end{equation}
When the entanglement in the energy eigenstates is a volume law, as is generically expected \cite{PageEntanglement,Vidmar_2017}, $S_2(E_j) \propto N$ and individual populations $\vert b_j \vert^2$ are exponentially small. This means that counterexamples such as Eq.~\eqref{eq:binom}, in which the individual populations are as large as $\propto N^{-1/2}$, are ruled out. 
\begin{figure}[t!]
\centering
\includegraphics[trim=0.02cm 0.2cm 0 0, clip,scale=0.55]{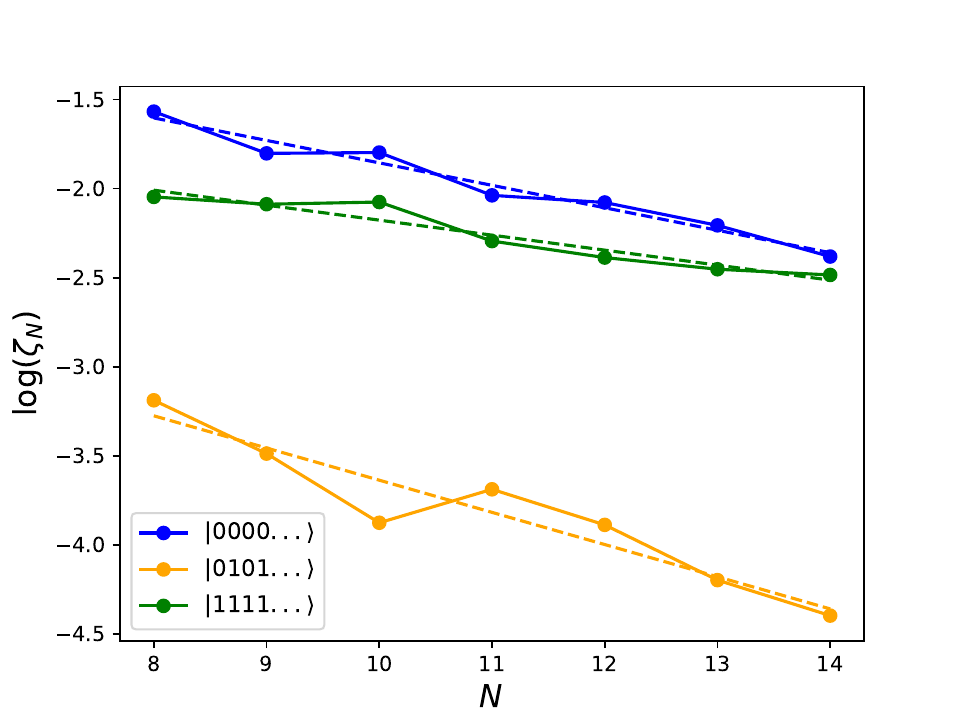}
\caption{\ksr{Illustration of the scaling of the logarithm of Berry-Esseen error $\zeta_N$ with the size of the system $N$ for the Mixed Field Ising Hamiltonian as defined in Eq.~\eqref{eq:mixed_field_ising} with the parameters $(J,g,h) = (1,-1.05,0.5)$.
An approximately linear decrease in $\log\zeta_N$ is observed even for finite system sizes (8 to 14 sites).}}
\label{fig:be_error_scaling}
\end{figure}
In these rather ubiquitous Hamiltonians, the product states thus have support in exponentially many eigenstates, and with exponentially small populations. This means that $J(x)$ in Def. \ref{def:Improved_BE} may be exponentially close to the smooth function $G(x)$, and in particular, we expect that the Berry-Esseen error is exponentially small in system size $\zeta_N = e^{-\Omega (N)}$. 
This is consistent with the numerical results in \cite{Banuls_2020}.
\ksr{In Fig.~\ref{fig:be_error_scaling} we plot the scaling of $\log\zeta_N$ with $N$ for the mixed field Ising model Hamiltonian defined as follows,
\begin{align}\label{eq:mixed_field_ising}
    H = \sum_{i=1}^{N-1}J\sigma^z_i\sigma^z_{i+1}+g\sigma^x_i+h\sigma^z_i,
\end{align}
where we chose the values of $(J,g,h) = (1,-1.05,0.5)$ for the non-integrable regime where ETH is expected to hold. 
For various choices of product states, we observe approximately linearly decreasing values for $\log\zeta_N$ with increasing $N$ even for finite system sizes, consistent with an exponential decay of $\zeta_N$ as theoretically expected from Eq.~\eqref{eq:state_populations}.
}

With such a favourable scaling, Lemma \ref{le:Variance_bound_general} allows one to increase the filter parameter \ksr{up to $M \propto e^{\Omega(N\log N)}$} in order to achieve vanishingly small variances \ksr{given Eq.~\eqref{eq:varlemma}. However in this case, the bond dimension (Eq.~\eqref{eq:BDBound}) restricts the minimum possible variance.} In particular, a variance of $\delta^2 \propto 1/\log N$ can be reached, which converges to zero in the thermodynamic limit, while still guaranteeing a polynomial bond dimension in Eq. \eqref{eq:BDBound}. 

It is also in these ETH Hamiltonians that we expect that states of low energy variance will have small subsystems resembling those of the Gibbs distribution. Considering $D \propto \text{poly}(N)$, and a corresponding entanglement entropy $\propto \log D$, subsystems of $\mathcal{O}(\log D)$ particles will display large amounts of entanglement entropy. Thus, it is possible that the filtered state displays approximately thermal marginals of up to that size. Along these lines, and depending on the observables considered, previous numerical analyses suggest one might need up to $\delta^2 \propto 1/\log^2 N$ \cite{Banuls_2020,Lu2021} (see in particular page 6 of \cite{Banuls_2020}). 
Additionally, the analytical results from \cite{Dymarsky2019} suggesting that for more general observables, variances vanishing as $\delta \propto 1/\text{poly}(N)$ might be needed.   


\section{Conclusion}

We have shown that for arbitrary systems evolving under local Hamiltonians, it is possible to efficiently construct states with variance $\delta = \Omega(1)$ via MPS representations, while shifting the average at most by a comparable amount. This is the smallest possible variance one could achieve, due to existing counterexamples. However, we expect that significantly lower values of the variance can be reached in most cases of interest, with the complexity of the algorithm increasing accordingly. In particular, in many models, $\delta=\Omega(1/\sqrt{\log N})$ can still be achieved efficiently.

With our main results, we provide rigorous analytical bounds on the classical simulability of the algorithm proposed in \cite{Banuls_2020,Lu2021}, which has recently been implemented with tensor networks \cite{Yilun2022}. We expect that this type of scheme will be important in near-term quantum simulation experiments of equilibrium and non-equilibrium properties of quantum many-body systems. Specifically, we show that to perform calculations beyond the reach of classical computers, these will have to go beyond one dimension or reach very small energy variances of at least $\delta=o(1/\sqrt{\log N})$. \alv{
These low energy-variance and MPS representable states are easy to prepare with quantum circuits. Thus they can serve as starting points of a filtering quantum algorithm, that will further reduce their variance and address quantities that are difficult to get with tensor network techniques. 
By starting with states with relatively low variance, one can conserve quantum computational resources. 
We hope that this further motivates experimental efforts on quantum simulation, such as \cite{hemery2023measuring}.}

The key to our method of analyzing the filtered state is the knowledge of the gaussianity of the wavefunction granted by the Berry-Esseen theorem. This allows us to, for instance, lower bound the norm of a product state to which an operator has been applied to (in this case, the filter $g_y$). \alv{It would be interesting to see if our proof techniques can be applied to other types of operators, such as the Chebyshev filter proposed in \cite{Banuls_2020}.}
We expect that this type of argument will have further applications in the study of classical and quantum algorithms for many body systems at finite energy density. 

\acknowledgements
The authors acknowledge useful discussions with Yilun Yang. 
AMA acknowledges support from the Alexander von Humboldt foundation and the Spanish Agencia Estatal de Investigacion through the grants ``IFT Centro de
Excelencia Severo Ochoa CEX2020-001007-S" and ``Ram\'on
y Cajal RyC2021-031610-I", financed by
MCIN/AEI/10.13039/501100011033 and the European
Union NextGenerationEU/PRTR. 
KSR acknowledges support from the European Union’s Horizon research and innovation programme through the ERC StG FINE-TEA-SQUAD (Grant No. 101040729). 
Views and opinions expressed are however those of the author(s) only and do not necessarily reflect those of the funding institutions. 
Neither of the funding institutions can be held responsible for them.
The research is part of the Munich Quantum Valley, which is supported by the Bavarian state
government with funds from the Hightech Agenda Bayern Plus. JIC acknowledges funding from the German Federal Ministry of Education and Research (BMBF) through EQUAHUMO (Grant No. 13N16066) within the funding program quantum technologies—from basic research to market.

\bibliographystyle{unsrtnat}
\bibliography{references}

\widetext
\appendix
\section{Proof of Lemma \ref{le:BEerror}}
\label{sec:Improved_BE_proof}
We study the characteristic function of the system's energy (assuming $\langle H \rangle=0$)
\begin{equation}
    \phi(t')\equiv \expval{e^{\frac{it'H}{\sigma}}} = \expval{e^{iHt}}.
\end{equation}
The bound on the error follows from Esseen's inequality \cite{feller1991introduction},
\begin{equation}\label{eq:esseen}
    \pi \zeta_N \leq \frac{C}{T} + \int\limits_0^T\frac{\vert\phi(t')-e^{-\frac{(t')^2}{2}}\vert}{t'}\text{d}t',
\end{equation}
where $ C \leq \frac{18}{\sqrt{2\pi}}$.
To estimate this, we just need the following Lemma bounding how close the characteristic function is to a Gaussian for small $t'$.
\begin{lemma}
Let $\ket{p}$ be a product state such that $\expval{H} \equiv \bra{p}H\ket{p}= 0$ and $\expval{H^2} = \sigma^2$. Moreover, let $t'\leq \frac{t^*\sigma}{2}$ with $t^*=\mathcal{O}(1)$. Then,
\begin{equation}
    \Big|\log\phi(t')-\left(\frac{-(t')^2}{2}\right)\Big|\leq 2\frac{N}{\sigma^3}\left(\frac{t'}{t^*}\right)^3.
\end{equation}
\end{lemma}
This follows straightforwardly from Theorem 10 in \cite{Wild_2022}, where $t^*$ is defined as an $\mathcal{O}(1)$ number that depends on the connectivity of the Hamiltonian. 
As a consequence, we have that, for $t'\leq \frac{t^*\sigma}{2}$, \ksr{
\begin{align}
    -2\frac{N}{\sigma^3}\left(\frac{t'}{t^*}\right)^3
    &\leq
    \log\phi(t')+\left(\frac{(t')^2}{2}\right)
    \leq 
    2\frac{N}{\sigma^3}\left(\frac{t'}{t^*}\right)^3,
\end{align}    
Subtracting the term $\frac{(t')^2}{2}$ from the whole equation, and then exponentiating, we get
\begin{align}
    e^{-\frac{(t')^2}{2}}\left(1-\mathcal{O}\left(\frac{N}{\sigma^3}\left(\frac{t'}{t^*}\right)^3\right)\right)
    \leq
    \phi(t')
    \leq
    e^{-\frac{(t')^2}{2}}\left(1+\mathcal{O}\left(\frac{N}{\sigma^3}\left(\frac{t'}{t^*}\right)^3\right)\right)
\end{align}
Subtracting $e^{-\frac{(t')^2}{2}}$ from both sides,
\begin{equation}
    \vert\phi(t') - e^{-\frac{t'^2}{2}}\vert 
    \leq
    e^{\frac{-(t')^2}{2}}\mathcal{O}
    \left(\frac{N}{\sigma^3}\left(\frac{t'}{t^*}\right)^3\right).
\end{equation}
Plugging this bound in Eq. \eqref{eq:esseen} and integrating yields,
\begin{align}
    \pi\zeta_N
    \leq
    \frac{C}{T}
    +
    \int\limits_0^T{
    e^{\frac{-(t')^2}{2}}\mathcal{O}
    \left(\frac{N}{\sigma^3}\frac{(t')^2}{(t^*)^3}\right)dt'}
\end{align}
Integrating and choosing $T = \frac{\sigma t^*}{2}$ and $t^*=\mathcal{O}(1)$,
\begin{equation}
    \pi \zeta_N \leq \frac{2C}{ \sigma t^*}+\frac{3N}{(\sigma t^*)^3}
    =
    \mathcal{O}\left(\frac{1}{\sqrt{N}}\right),
\end{equation}
using Assumption~\ref{ass1}, the result follows.}


\section{Proof of Lemma \ref{le:Variance_bound_general}}\label{app:proofMain}
To have a truncated number of terms in the final bond dimension calculation we use the approximate cosine filter, $g_y\left(\frac{H-E}{N}\right) = \frac{1}{2^{M}}\sum\limits_{m=-y\sqrt{M}}^{y\sqrt{M}}\binom{M}{\frac{M}{2}-m}e^{i2m(\frac{H-E}{N})}$, where $E$ is the average energy around which we filter. 
Applying this operator to the initial state $\ket{p}$, the \ksr{normalized} filtered state $\ket{p_{M,y}}$ can be written as
\begin{equation}
    \ket{p_{M,y}}= \frac{g_y\left(\frac{H-E}{N}\right) \ket{p}}{\norm{g_y\left(\frac{H-E}{N}\right) \ket{p}}}.
\end{equation}
Let the variance of the initial state $\ket{p}$ be $\sigma = \mathcal{O}(\sqrt{N})$. Given the Hamiltonian operator $H$, the new variance is defined as
\begin{equation}
    \delta^2 = \bra{p_{M,y}}H^2\ket{p_{M,y}} - \bra{p_{M,y}}H\ket{p_{M,y}}^2,
\end{equation} 
where $\mu = \bra{p_{M,y}}H\ket{p_{M,y}}$ is the average of the filtered state. 

The average is not expected to change significantly from $E$ on application of the filter operator since both the filter and the initial state are taken to be centered around the energy $E$ and the filter is symmetric about the average. We first prove this by showing that $\vert \mu - E \vert = \mathcal{O}\left(\frac{N}{\sqrt{M}}\right)$, and then upper bound the variance in a similar manner. 
Without loss of generality, we choose $E = 0$. We begin by writing the initial state $\ket{p}$ in the energy eigenbasis of Hamiltonian $H$,
\begin{equation}
    \ket{p} = \sum\limits_{E_j}b_j\ket{E_j},
\end{equation}
where $E_j$'s represent the eigenvalues.
\ksr{We now calculate the average $\mu$ by substituting $\ket{p}$ from the above equation}
\begin{equation}\label{eq:av_step1}
    \mu = \bra{p_{M,y}}H\ket{p_{M,y}} 
    =\frac{\bra{p}g_y^2\left(\frac{H}{N}\right)H\ket{p}}{\norm{g_y\left(\frac{H-E}{N}\right) \ket{p}}^2}
    =\frac{\sum\limits_{E_j}\vert b_j\vert^2g_y^2\left(\frac{E_j}{N}\right)E_j}
    {\sum\limits_{E_j}\vert b_j\vert^2g_y^2\left(\frac{E_j}{N}\right)}.
\end{equation}
\ksr{Recall that the approximate cosine filter was introduced (Lemma~\ref{def:grouped_cosine}) to decrease the extensive number of terms in the time evolution representation of the original cosine filter.}
The approximate cosine filter is related to the original filter in terms of the following inequalities
\begin{equation}\label{eq:ineq_grouped_cosine}
    \cos^{2M}\left(\frac{H}{N}\right) - 4e^{-\frac{y^2}{2}} \leq g_y^2\left(\frac{H}{N}\right)\leq \cos^{2M}\left(\frac{H}{N}\right) + 4e^{-y^2}.
\end{equation}
\ksr{In the upcoming sections, we will use the setting introduced above to bound the average and variance of the final filtered state.}
\subsection{Average energy}

The distance from the mean $\mu$ can be bounded as follows
\ksr{by substituting equation~\eqref{eq:ineq_grouped_cosine} into equation~\eqref{eq:av_step1},}
\begin{align}\label{eq:main_mean_eqn1}
\begin{split}
   \vert \mu  \vert
    \leq 
   \left \vert \frac{\sum\limits_{E_{j}}{\vert b_{j}\vert^2\cos^{2M}{\left(\frac{E_j}{N}\right)}E_j}+4e^{-y^2}\sum\limits_{E_j}\vert b_j\vert^2\vert E_j\vert}{\sum\limits_{E_{j}}{\vert b_{j}\vert^2\cos^{2M}{\left(\frac{E_j}{N}\right)}}-4e^{-\frac{y^2}{2}}\sum\limits_{E_j}{|b_j|^2}} 
   \right \vert
    &\leq  
    \left \vert
    \frac{\sum\limits_{E_{j}}{|b_{j}|^2\cos^{2M}{\left(\frac{E_j}{N}\right)}E_j}+4e^{-y^2}\sum\limits_{E_j}\vert b_j\vert^2\vert E_j\vert}{\sum\limits_{E_{j}}{|b_{j}|^2\cos^{2M}{\left(\frac{E_j}{N}\right)}}-4e^{-\frac{y^2}{2}}} \right \vert
    \\
    &:= \left \vert \frac{\mathcal{M}_{1(+)}+\mathcal{M}_{1(-)}+\mathcal{E}_1}{\mathcal{M}_2} \right \vert,
\end{split}
\end{align}
where the second inequality uses that the initial average energy is $0$ and $\sum\limits_{E_j}|b_j|^2 = 1$, and $\mathcal{M}_{1(+)}$, $\mathcal{M}_{1(-)}$, $\mathcal{E}_1$, and $\mathcal{M}_2$ are defined as follows
\begin{align}
    \mathcal{M}_{1(+)} &= \sum\limits_{E_{j}>0}{|b_{j}|^2\cos^{2M}{\left(\frac{E_j}{N}\right)}E_j},
    \\
    \mathcal{M}_{1(-)} &= \sum\limits_{E_{j}<0}{|b_{j}|^2\cos^{2M}{\left(\frac{E_j}{N}\right)}E_j},
    \\
    \mathcal{E}_1 &= 
    4e^{-y^2}\sum\limits_{E_j}\vert b_j\vert^2\vert E_j\vert,
    \\
    \mathcal{M}_2 &= \sum\limits_{E_{j}}{|b_{j}|^2\cos^{2M}{\left(\frac{E_j}{N}\right)}}-4e^{-\frac{y^2}{2}}.
\end{align}
First we focus on simplifying the numerator of Eq. \eqref{eq:main_mean_eqn1}. 
We use what we term the \textit{segmentation method} (motivated from a similar procedure used in \cite{Hovhannisyan_2021}). \alv{As such, the initial steps of the proof mirror those of Appendix E in \cite{Hovhannisyan_2021}, although here the object we are estimating yields a variety of additional technical complications, due to the presence of the filter, and require more careful error estimates.} 
The start of the method is to divide the energy range $E_{\text{max}}\sim \mathcal{O}(N)$ into $R_1$ number of segments of width $\Lambda_1> 0$ each, so that 
\begin{equation}\label{eq:segmentation_method}
    \mathcal{M}_{1(+)}
    =
    \sum\limits_{E_{j}>0}{|b_j|^2\cos^{2M}{\left(\frac{E_j}{N}\right)}E_j} 
    \leq \sum\limits_{l=0}^{R_1-1}{\cos^{2M}{\left(\frac{\Lambda_1l}{N}\right)}(\Lambda_1(l+1))}\sum\limits_{\Lambda_1 l<E_j<\Lambda_1(l+1)}{|b_j|^2}.
\end{equation}
Using Definition \ref{def:Improved_BE}, the coefficients $b_j$ can be approximated using a Gaussian integral
\begin{equation}
    \sum\limits_{\Lambda_1 l<E_j<\Lambda_1(l+1)}{|b_j|^2}
    \leq
    \int\limits_{\Lambda_1l}^{\Lambda_1(l+1)}{\frac{dt}{\sigma\sqrt{2\pi}}e^{-\frac{t^2}{2\sigma^{2}}}}+\zeta_N.
\end{equation}
Since by definition $\mathcal{M}_{1(-)}\le 0$, the terms in the numerator can be bounded as
\begin{equation}
     \mathcal{M}_{1(+)}
     \leq 
     \sum\limits_{l=0}^{R_1-1}{\cos^{2M}{\left(\frac{\Lambda_1l}{N}\right)}(\Lambda_1(l+1))\left[\int\limits_{\Lambda_1l}^{\Lambda_1(l+1)}{\frac{dt}{\sigma\sqrt{2\pi}}e^{-\frac{t^2}{2\sigma^{2}}}}+\zeta_N\right]},
\end{equation}
\begin{equation}
     -\mathcal{M}_{1(-)}
     \leq 
     \sum\limits_{l=0}^{R_1-1}{\cos^{2M}{\left(\frac{\Lambda_1l}{N}\right)}(\Lambda_1(l+1))\left[\int\limits_{-\Lambda_1(l+1)}^{-\Lambda_1l}{\frac{dt}{\sigma\sqrt{2\pi}}e^{-\frac{t^2}{2\sigma^{2}}}}+\zeta_N\right]},
\end{equation}
Let $\tilde{t} = \frac{t}{\sqrt{2}\sigma},  d\tilde{t} = \frac{dt}{\sqrt{2}\sigma}$, and $\tilde{\Lambda}_1 = \frac{\Lambda_1}{\sqrt{2}\sigma}$, then $\mathcal{M}_{1(+)}$ can be bounded as
\begin{align}
\begin{split}
    \mathcal{M}_{1(+)}
   \leq 
   \frac{\sqrt{2}\sigma}{\sqrt{\pi}}\sum\limits_{l=0}^{R_1-1}&{\int\limits_{\tilde{\Lambda}_1 l}^{\tilde{\Lambda}_1(l+1)}{d\tilde{t}\cos^{2M}{\left(\sqrt{2}\sigma\left(\frac{\tilde{t}-\tilde{\Lambda}_1}{N}\right)\right)}(\tilde{t}+\tilde{\Lambda}_1) e^{-\tilde{t}^{2}}}}+
   \\&+
   \sqrt{2}\sigma\zeta_N\sum\limits_{l=0}^{R_1-1}{\cos^{2M}{\left(\frac{\sqrt{2}\sigma\tilde{\Lambda}_1l}{N}\right)}\tilde{\Lambda}_1(l+1)}
\end{split}
\end{align}
\begin{align}
\begin{split}
   \Rightarrow\mathcal{M}_{1(+)}
   \leq 
   \frac{\sqrt{2}\sigma}{\sqrt{\pi}}&\int\limits_0^{\tilde{\Lambda}_1 R_1}{d\tilde{t}\cos^{2M}{\left(\sqrt{2}\sigma\left(\frac{\tilde{t}-\tilde{\Lambda}_1}{N}\right)\right)}(\tilde{t}+\tilde{\Lambda}_1) e^{-\tilde{t}^{2}}}+
    \\&+
    \sqrt{2}\sigma\zeta_N\sum\limits_{l=0}^{R_1-1}{\cos^{2M}{\left(\frac{\sqrt{2}\sigma\tilde{\Lambda}_1 l}{N}\right)}(\tilde{\Lambda}_1(l+1))}.
\end{split}
\end{align}
Similarly
\begin{align}
\begin{split}
    - \mathcal{M}_{1(-)}
   \leq 
    -\frac{\sqrt{2}\sigma}{\sqrt{\pi}}&\int\limits_{-\tilde{\Lambda}_1 R_1}^0{d\tilde{t}\cos^{2M}{\left(\sqrt{2}\sigma\left(\frac{\tilde{t}+\tilde{\Lambda}_1}{N}\right)\right)}(\tilde{t}-\tilde{\Lambda}_1) e^{-\tilde{t}^{2}}}+
    \\&+\sqrt{2}\sigma\zeta_N\sum\limits_{l=0}^{R_1-1}{\cos^{2M}{\left(\frac{\sqrt{2}\sigma\tilde{\Lambda}_1 l}{N}\right)}(\tilde{\Lambda}_1(l+1))}.
\end{split}
\end{align}

We get the following upper bound on the $\vert  \mathcal{M}_{1(+)} + \mathcal{M}_{1(-)} 
 \vert$, 
\begin{align}\label{eq:main_numerator_avg_energy}
\begin{split}
    \vert \mathcal{M}_{1(+)} + \mathcal{M}_{1(-)}  \vert
   \leq 
    \frac{2\sqrt{2}\sigma}{\sqrt{\pi}}\int\limits_0^{\tilde{\Lambda}_1 R_1}d\tilde{t}e^{-\tilde{t}^2}\left[
    \cos^{2M}{\left(\sqrt{2}\sigma\left(\frac{\tilde{t}-\tilde{\Lambda}_1}{N}\right)\right)}\tilde{\Lambda}_1
    \right]+
    \\+
    2\sqrt{2}\sigma\zeta_N\sum\limits_{l=0}^{R_1-1}{\cos^{2M}{\left(\frac{\sqrt{2}\sigma\tilde{\Lambda}_1 l}{N}\right)}(\tilde{\Lambda}_1(l+1))}.
\end{split}
\end{align}
Additionally, the error term $\mathcal{E}_1$ can be bounded similarly as
\begin{align}
    \mathcal{E}_1
    &\leq
    4e^{-y^2}\sum_{l=0}^{R_1-1}\Lambda_1(l+1)\sum_{\Lambda_1l\leq E_j\leq\Lambda_1(l+1)}\vert b_j\vert^2
    \\
    &\leq
    4e^{-y^2}\left[
    \frac{\sqrt{2}\sigma}{\sqrt{\pi}}
    \int_{0}^{\tilde{\Lambda}_1R_1}d\tilde{t}e^{-\tilde{t}^2}(\tilde{t}+\tilde{\Lambda}_1)
    +
    \sqrt{2}\sigma\zeta_N
    \sum_{l=0}^{R_1-1}\tilde{\Lambda}_1(l+1)
    \right]
    \\
    &\leq
    2\sqrt{2}\sigma e^{-y^2}\left[
    \frac{(1+\sqrt{\pi}\tilde{\Lambda}_1)}{\sqrt{\pi}}
    +
    \zeta_N\tilde{\Lambda}_1(R_1^2+3R_1)
    \right]
    \\
    &\leq
    2\sqrt{2}\sigma e^{-y^2}\left[
    \frac{(1+\sqrt{\pi}\tilde{\Lambda}_1)}{\sqrt{\pi}}
    +
    4\zeta_N\sqrt{N}R_1
    \right]
    .\label{eq:err_term_average_num}
\end{align}

Applying a similar method to lower bound the first term in the denominator of Eq. \eqref{eq:main_mean_eqn1},
\begin{align}
    \sum\limits_{E_j}{|b_j|^2\cos^{2M}{\left(\frac{E_j}{N}\right)}} 
    &\geq 
    2\sum\limits_{l=0}^{R_2-1}{\cos^{2M}{\left(\frac{\Lambda_2(l+1)}{N}\right)}}\sum\limits_{\Lambda_2 l<E_j<\Lambda_2 (l+1)}{|b_j|^2}
    \\
   &\geq 
   2\sum\limits_{l=0}^{R_2-1}{\cos^{2M}{\left(\frac{\Lambda_2(l+1)}{N}\right)}\left[\int\limits_{\Lambda_2 l}^{\Lambda_2 (l+1)}{\frac{dt}{\sigma\sqrt{2\pi}}e^{-\frac{t^2}{2\sigma^{2}}}}-\zeta_N\right]}
   \\
   &\geq 
   \frac{2}{\sqrt{\pi}}\sum\limits_{l=0}^{R_2-1}{\int\limits_{\tilde{\Lambda}_2 l}^{\tilde{\Lambda}_2 (l+1)}{d\tilde{t}\cos^{2M}{\left(\sqrt{2}\sigma\left(\frac{\tilde{t}+\tilde{\Lambda}_2}{N}\right)\right)}e^{-\tilde{t}^{2}}}}\nonumber
   \\&-
2\zeta_N\sum\limits_{l=0}^{R_2-1}{\cos^{2M}{\left(\frac{\sqrt{2}\sigma\tilde{\Lambda}_2(l+1)}{N}\right)}}.
\end{align}
We get the following lower bound on the denominator
\begin{align}\label{eq:AppendixDenominator}
   \mathcal{M}_2
   \geq 
   \frac{2}{\sqrt{\pi}}\int\limits_{0}^{\tilde{\Lambda}_2R_2}{d\tilde{t}\cos^{2M}{\left(\sqrt{2}\sigma\left(\frac{\tilde{t}+\tilde{\Lambda}_2}{N}\right)\right)}e^{-\tilde{t}^{2}}}
    -2\zeta_N\sum\limits_{l=0}^{R_2-1}{\cos^{2M}{\left(\frac{\sqrt{2}\sigma\tilde{\Lambda}_2(l+1)}{N}\right)}}-4e^{-\frac{y^2}{2}}.
\end{align}
Combining the bounds on the numerator \ksr{(Eq.~\eqref{eq:main_numerator_avg_energy} and Eq.~\eqref{eq:err_term_average_num})} and the denominator \ksr{(Eq.~\eqref{eq:AppendixDenominator})}
\begin{align}
   \vert \mu  \vert 
    \leq 
    \frac{\sqrt{2}\sigma\left[\frac{1}{\sqrt{\pi}}\int\limits_0^{\tilde{\Lambda}_1 R_1}d\tilde{t}e^{-\tilde{t}^2}\left[
    \cos^{2M}{\left(\sqrt{2}\sigma\left(\frac{\tilde{t}-\tilde{\Lambda}_1}{N}\right)\right)}\tilde{\Lambda}_1
    \right]
    +
    \mathcal{U}\right]
    }
    {\frac{1}{\sqrt{\pi}}\int\limits_{0}^{\tilde{\Lambda}_2R_2}{d\tilde{t}\cos^{2M}{\left(\sqrt{2}\sigma\left(\frac{\tilde{t}+\tilde{\Lambda}_2}{N}\right)\right)}e^{-\tilde{t}^{2}}}
   -\zeta_N\sum\limits_{l=0}^{R_2-1}{\cos^{2M}{\left(\frac{\sqrt{2}\sigma\tilde{\Lambda}_2(l+1)}{N}\right)}}
   -2e^{-\frac{y^2}{2}}},
\end{align}
where 
\begin{align}
    \mathcal{U}=\zeta_N\sum\limits_{l=0}^{R_1-1}{\cos^{2M}{\left(\frac{\sqrt{2}\sigma\tilde{\Lambda}_1 l}{N}\right)}(\tilde{\Lambda}_1(l+1))}
    +
    2 e^{-y^2}\left[
    \frac{(1+\sqrt{\pi}\tilde{\Lambda}_1)}{\sqrt{\pi}}
    +
    4\zeta_N\sqrt{N}R_1
    \right].
\end{align}
Note the following upper and lower bounds on $\cos^a(x)$ for $a>0$,
\begin{equation}
     e^{-ax^2}\leq \cos^a(x) \leq e^{-\frac{ax^2}{2}}.
    \label{eq:cosine-property2}
\end{equation}
We can use these to replace the cosine power function in the numerator and denominator to obtain 
\begin{equation}
     \vert \mu  \vert  \leq \frac{
     \begin{split}
    \frac{\sqrt{2}\sigma\tilde{\Lambda}_1}{\sqrt{\pi}}\int\limits_{0}^{\tilde{\Lambda}_1 R_1}{d\tilde{t} e^{-2M\sigma^2(\frac{\tilde{t}-\tilde{\Lambda}_1}{N})^{2}} e^{-\tilde{t}^{2}}}+\sqrt{2}\sigma&\zeta_N\sum\limits_{l=0}^{R_1-1}{e^{-2M\sigma^2\frac{\tilde{\Lambda}_1^2l^2}{N^2}}(\tilde{\Lambda}_1(l+1))}+
    \\&+2\sqrt{2}\sigma e^{-y^2}\left[
    \frac{(1+\sqrt{\pi}\tilde{\Lambda}_1)}{\sqrt{\pi}}
    +
    4\zeta_N\sqrt{N}R_1
    \right]
     \end{split}
     }
    {\frac{1}{\sqrt{\pi}}\int\limits_{0}^{\tilde{\Lambda}_2R_2}{d\tilde{t}e^{-4M\sigma^2(\frac{\tilde{t}+\tilde{\Lambda}_2}{N})^{2}}e^{-\tilde{t}^{2}}}-\zeta_N\sum\limits_{l=0}^{R_2-1}{e^{-4M\sigma^2(\frac{\tilde{\Lambda}_2(l+1)}{N})^{2}}}-4e^{-\frac{y^2}{2}}}.
\end{equation}
Now \ksr{using Assumption~\ref{ass1} by} substituting $\sigma = s\sqrt{N}$,
\begin{equation}\label{eq:main_mean_eqn2}
    \vert\mu\vert 
    \leq 
    \frac{
    \begin{split}
    \sqrt{2N}s\Bigg[
    \frac{\tilde{\Lambda}_1}{\sqrt{\pi}}\int\limits_{0}^{\tilde{\Lambda}_1 R_1}{d\tilde{t} e^{-\frac{2s^2M}{N}(\tilde{t}-\tilde{\Lambda}_1)^{2}} e^{-\tilde{t}^{2}}}
    +
    \zeta_N&\sum\limits_{l=0}^{R_1-1}{e^{-\frac{2s^2M}{N}\tilde{\Lambda}_1^2l^2}(\tilde{\Lambda}_1(l+1))}+
    \\&+2 e^{-y^2}\left[
    \frac{(1+\sqrt{\pi}\tilde{\Lambda}_1)}{\sqrt{\pi}}
    +
    4\zeta_N\sqrt{N}R_1
    \right]\Bigg]
    \end{split}
    }
    {\frac{1}{\sqrt{\pi}}\int\limits_{0}^{\tilde{\Lambda}_2R_2}{d\tilde{t}e^{-\frac{4s^2M}{N}(\tilde{t}+\tilde{\Lambda}_2)^{2}}e^{-\tilde{t}^{2}}}-\zeta_N\sum\limits_{l=0}^{R_2-1}{e^{-\frac{4s^2M}{N}(\tilde{\Lambda}_2(l+1))^{2}}}-4e^{-\frac{y^2}{2}}} 
    \\:= \frac{\mathcal{M}'_1}{\mathcal{M}'_2}.
\end{equation}
The first term of $\mathcal{M}'_1$ can be bounded as
\begin{align}\label{eq:numdem}
\begin{split}
    \frac{\sqrt{2N}s\tilde{\Lambda}_1}{\sqrt{\pi}}\int\limits_{0}^{\tilde{\Lambda}_1 R_1} &{d\tilde{t} e^{-\frac{2s^2M}{N}(\tilde{t}-\tilde{\Lambda}_1)^{2}} e^{-\tilde{t}^{2}}} \\
    &\leq
    \frac{\sqrt{N}s\tilde{\Lambda}_1e^{-\frac{\tilde{\Lambda}_1^2}{(1+\frac{N}{2s^2M})}}}{\sqrt{2}\sqrt{1+\frac{2s^2M}{N}}}
    \Bigg[
    \erf\left(\tilde{\Lambda}_1R_1\sqrt{1+\frac{2s^2M}{N}}-\frac{2s^2M\tilde{\Lambda}_1}{N}\frac{1}{\sqrt{1+\frac{2s^2M}{N}}}\right)
    \\&+
    \erf\left(\frac{2s^2M\tilde{\Lambda}_1}{N}\frac{1}{\sqrt{1+\frac{2s^2M}{N}}}\right)
    \bigg]
    \\
    &\leq 
    \frac{\sqrt{2N}s\tilde{\Lambda}_1e^{-\frac{\tilde{\Lambda}_1^2}{1+\frac{N}{2s^2M}}}}{\sqrt{1+\frac{2s^2M}{N}}}.
\end{split}
\end{align}
In the second inequality, we have used that the error function is bounded by 1. To bound the discrete sum error terms, we use the\textit{ Euler-Maclaurin formula} to change summations to integrals. It is given as follows
\begin{align}\label{eq:euler_maclaurin}
\sum_{i=a}^bf(i)-\int_a^bf(x)dx
=
\frac{f(a)+f(b)}{2}
+
\sum_{k=1}^{\lfloor\frac{p}{2}\rfloor}\frac{B_{2k}}{(2k)!}\left(f^{(2k-1)}(b)-f^{(2k-1)}(a)\right)+\mathcal{R}_p,
\end{align}
where $B_n$ denotes the $n^{\text{th}}$ Bernoulli number, and the remainder term $\mathcal{R}_p$ is bounded as
\begin{equation}
    \vert\mathcal{R}_p\vert
    \leq
    \frac{2\zeta(p)}{(2\pi)^p}\int_a^b\vert f^{(p)}(x)\vert dx.
\end{equation}
The choice of $p$ is up to us. Taking $p=2$, the error term of $\mathcal{M}'_1$ can be bounded as
\begin{align}\label{eq:mean_num_error_term}
\begin{split}
   \sum\limits_{l=0}^{R_1-1}{e^{-\frac{2s^2M}{N}\tilde{\Lambda}_1^2l^2}(\tilde{\Lambda}_1(l+1))} 
    &\leq 
   \int\limits_0^{R_1} dl e^{-\frac{2s^2M}{N}\tilde{\Lambda}_1^2l^2}(\tilde{\Lambda}_1(l+1))
    +
    \frac{\tilde{\Lambda}_1+\tilde{\Lambda}_1R_1 e^{-\frac{2s^2M}{N}\tilde{\Lambda}_1^2(R_1-1)^2}}{2}
    +
    \\
    &+
    \frac{B_2}{2}\tilde{\Lambda}_1\left[e^{-\frac{2s^2M}{N}\tilde{\Lambda}_1^2(R_1-1)^2}\left(1-\frac{4s^2M}{N}\tilde{\Lambda}_1^2R_1(R_1-1)\right)-1\right]
    +
    \mathcal{R}_2.
\end{split}
\end{align}
The integral on the RHS can be bounded as follows
\begin{align}
    \int\limits_0^{R_1} dl e^{-\frac{2s^2M}{N}\tilde{\Lambda}_1^2l^2}(\tilde{\Lambda}_1(l+1))
    &\leq
    \frac{N}{4\tilde{\Lambda}_1Ms^2}(1-e^{-\frac{2s^2M}{N}\tilde{\Lambda}_1^2R_1^2})
    +
    \frac{\sqrt{\pi}}{2\sqrt{2}s}\sqrt{\frac{N}{M}}\erf\left(2\tilde{\Lambda}_1R_1s\sqrt{\frac{M}{N}}\right),
    \\
    &\leq
    \frac{N}{4\tilde{\Lambda}_1Ms^2}
    +
    \frac{\sqrt{\pi}}{2\sqrt{2}s}\sqrt{\frac{N}{M}}.
\end{align}
The the sum in Eq. \eqref{eq:mean_num_error_term} can be bounded as
\begin{align}
\begin{split}
   \sum\limits_{l=0}^{R_1-1}{e^{-\frac{2s^2M}{N}\tilde{\Lambda}_1^2l^2}(\tilde{\Lambda}_1(l+1))} 
   \leq 
   \frac{N}{4\tilde{\Lambda}_1Ms^2}
   & +
    \frac{\sqrt{\pi}}{2\sqrt{2}s}\sqrt{\frac{N}{M}}
    +
    \frac{\tilde{\Lambda}_1+\tilde{\Lambda}_1R_1 e^{-\frac{2s^2M}{N}\tilde{\Lambda}_1^2(R_1-1)^2}}{2}+
   \\ & +
    \frac{\tilde{\Lambda}_1e^{-\frac{2s^2M}{N}\tilde{\Lambda}_1^2(R_1-1)^2}}{12}
    +
    \mathcal{R}_2,
\end{split}
\end{align}
where we substitute $B_2 = \frac{1}{6}$ for the second Bernoulli number and $\mathcal{R}_2$ is the remainder term from the approximation, bounded as follows
\begin{equation}
    \mathcal{R}_2 
    \leq
    \frac{\tilde{\Lambda}_1\zeta(2)}{2\pi^2}\left(5+\sqrt{\frac{8\pi s^2\tilde{\Lambda}_1^2M}{N}}\right).
\end{equation}
This means that
\begin{align}\label{eq:err_term_avg_num_simplified}
\begin{split}
    \sum\limits_{l=0}^{R_1-1}{e^{-\frac{2s^2M}{N}\tilde{\Lambda}_1^2l^2}(\tilde{\Lambda}_1(l+1))} 
    \leq 
    \frac{N}{4\tilde{\Lambda}_1Ms^2}
    +
    \frac{\sqrt{\pi}}{2\sqrt{2}s}\sqrt{\frac{N}{M}}
    +
    \frac{\tilde{\Lambda}_1+\tilde{\Lambda}_1(R_1+\frac{1}{6}) e^{-\frac{2s^2M}{N}\tilde{\Lambda}_1^2(R_1-1)^2}}{2}
    +
    \\
    +
    \frac{\tilde{\Lambda}_1\zeta(2)}{2\pi^2}\left(5+\sqrt{\frac{8\pi s^2\tilde{\Lambda}_1^2M}{N}}\right).
\end{split}
\end{align}
\ksr{The value of the constant $C_1$ is arbitrary. To achieve the best possible upper bound on the variance, we find that the best choice is} $C_1 := \frac{2s^2M}{N}\tilde{\Lambda}_1^2$. Then we can bound $\mathcal{M}'_1$ \ksr{by adding Eq.~\eqref{eq:numdem} and Eq.~\eqref{eq:err_term_avg_num_simplified}} as
\begin{align}
\begin{split}
    \mathcal{M}'_1
    \leq 
    \frac{\sqrt{2N}s\tilde{\Lambda}_1e^{-\frac{\tilde{\Lambda}_1^2}{(1+\frac{N}{2s^2M})}}}{\sqrt{1+\frac{2s^2M}{N}}}
    +
    \sqrt{2N}s\zeta_N\Bigg[
    \frac{\tilde{\Lambda}_1}{2C_1}\left(1+\sqrt{\pi C_1}\right)
    +
    \frac{\tilde{\Lambda}_1+\tilde{\Lambda}_1(R_1+\frac{1}{6}) e^{-C_1(R_1-1)^2}}{2}+
    \\+
    \frac{\tilde{\Lambda}_1\zeta(2)}{2\pi^2}\left(5+2\sqrt{\pi C_1}\right)
    \Bigg]
    \\+2\sqrt{2N}se^{-y^2}\zeta_N\sqrt{N}R_1+2s\sqrt{N}e^{-y^2}\frac{(1+\sqrt{\pi}\tilde{\Lambda}_1)}{\sqrt{\pi}}.
\end{split}
\end{align}
Using $\tilde{\Lambda}_1 = \frac{\Lambda_1}{s\sqrt{2N}}$, we can write the bound as follow
\begin{align}
\begin{split}
    \mathcal{M}'_1
    \leq 
    \sqrt{\frac{N}{2s^2M}}\frac{\Lambda_1e^{-\frac{1}{2s^2}\frac{\Lambda_1^2}{N}\frac{1}{1+\frac{N}{2s^2M}}}}{\sqrt{1+\frac{N}{2s^2M}}}
    +
    \zeta_N\Bigg[
    \frac{\Lambda_1}{2C_1}\left(1+\sqrt{\pi C_1}\right)
    +
    \frac{\Lambda_1+\Lambda_1(R_1+\frac{1}{6}) e^{-C_1(R_1-1)^2}}{2}+
    \\+
    \frac{\Lambda_1\zeta(2)}{2\pi^2}\left(5+2\sqrt{\pi C_1}\right)
    \Bigg]+
    \\+2\sqrt{2N}se^{-y^2}\zeta_N\sqrt{N}R_1+2s\sqrt{N}e^{-y^2}\frac{(1+\sqrt{\pi}\tilde{\Lambda}_1)}{\sqrt{\pi}}.
\end{split}
\end{align}
Note that $e^{-\frac{1}{2s^2}\frac{\Lambda_1^2}{N}\frac{1}{1+\frac{N}{2s^2M}}} \leq 1$ and $(1+\frac{N}{2s^2M})^{-1/2}\leq 1$ for any value of the free parameters and this is particularly tight approximation for $\frac{N}{2s^2M}<1$. Thus, choosing $\Lambda_1$ such that $C_1 = \mathcal{O}(1)$, we get the following overall upper bounds on the numerator of the average
\begin{align}
    \mathcal{M}'_1
    \leq 
    \frac{1}{s}\sqrt{\frac{N}{2M}}
    \Lambda_1
    +2\sqrt{2}se^{-y^2}N\zeta_NR_1+2s\sqrt{N}e^{-y^2}\frac{(1+\sqrt{\pi}\tilde{\Lambda}_1)}{\sqrt{\pi}}
    +
    \mathcal{O}(\zeta_N\Lambda_1).
\end{align}
For the choice of $y$, we require the following condition, such that the 2nd term is smaller than the leading term
\begin{align}
    e^{-y^2}&\leq\frac{1}{4s^2}\sqrt{\frac{\Lambda_1^2}{2MN\zeta_N^2R_1^2}}.
\end{align}
First substituting $R = \frac{N}{\Lambda_1}$ and then $ \Lambda_1 = \sqrt{C_1}\frac{N}{\sqrt{M}}$, we get the bound
\begin{align}
    y
    =\Omega\left(\sqrt{\log\frac{M^3\zeta_N^2}{N}}\right).
\end{align}
We now use Lemma \eqref{le:BEerror} to find the worst case bound on $y$ which will be applicable for all values of $\zeta_N$,
\begin{align}\label{eq:y_bound}
    y=\Omega\left(\sqrt{\log\frac{M^3}{N^2}}\right).
\end{align}
\ksr{The parameter $y$ dictates how many terms in Eq.~\eqref{eq:gyx-def} are to be included in the later calculation of the bond dimension.}
Using that $C_1 = \mathcal{O}(1)$, we get the scalings $\Lambda_1 \propto \frac{N}{\sqrt{M}}$ and $R_1 \propto \sqrt{M}$. Considering these, we choose $ y=\Omega\left(\sqrt{\log\frac{M^3\zeta_N^2}{N}}\right)$, to get the following upper bound on the numerator
\begin{align}
    \mathcal{M}'_1
    \leq 
    \frac{1}{s}\sqrt{\frac{N}{2M}}
    \Lambda_1
    +
    \mathcal{O}(\zeta_N\Lambda_1).
\end{align}
We now simplify the denominator of Eq. \eqref{eq:main_mean_eqn2}.  Define the variable
\begin{equation}
    C_2:=\frac{4s^2M}{N}\tilde{\Lambda}_2^2 = \frac{2M\Lambda_2^2}{N^2},
\end{equation}
\ksr{Using this value of $C_2$ and the Euler-Maclaurin formula from the equation~\eqref{eq:euler_maclaurin},} the integral in the first term of denominator $\mathcal{M}'_2$ can be lower bounded as
\begin{align}
\begin{split}
    \frac{1}{\sqrt{\pi}}\int\limits_{0}^{\tilde{\Lambda}_2R_2}{d\tilde{t}e^{-\frac{4s^2M}{N}(\tilde{t}+\tilde{\Lambda}_2)^{2}}e^{-\tilde{t}^{2}}}
    &\geq
    \frac{1}{2}
    \frac{e^{-\frac{4s^2M}{N}\frac{\tilde{\Lambda}_2^2}{1+\frac{4s^2M}{N}}}}{\sqrt{1+\frac{4s^2M}{N}}}
    \left[
    \erf\left(
    \tilde{\Lambda}_2R_2\sqrt{1+\frac{4s^2M}{N}}
    +
    \frac{\tilde{\Lambda}_2}{\sqrt{\frac{N}{4s^2M}}\sqrt{1+\frac{N}{4s^2M}}}
    \right)
    \right.
    \\
    &\quad
    \left.
    \quad\quad\quad\quad\quad\quad\quad\quad -
    \erf\left(
    \frac{\tilde{\Lambda}_2}{\sqrt{\frac{N}{4s^2M}}\sqrt{1+\frac{N}{4s^2M}}}
    \right)
    \right]
\end{split}
\end{align}
Since we have the freedom of choosing $\tilde{\Lambda}_2$, we take $C_2 = \frac{3}{4}$ (or $\tilde{\Lambda}_2 = \frac{3N}{16s^2M}$). Then the integral is bounded as
\begin{align}
    \frac{1}{\sqrt{\pi}}\int\limits_{0}^{\tilde{\Lambda}_2R_2}{d\tilde{t}e^{-\frac{4s^2M}{N}(\tilde{t}+\tilde{\Lambda}_2)^{2}}e^{-\tilde{t}^{2}}}
    &\geq
    \frac{1}{2}
    \frac{\tilde{\Lambda}_2e^{-\frac{3\tilde{\Lambda}_2^2}{4\tilde{\Lambda}_2^2+3}}}{\sqrt{\tilde{\Lambda}_2^2+\frac{3}{4}}}
    \left[
    \erf\left(
    R_2\sqrt{\tilde{\Lambda}_2^2+\frac{3}{4}}
    +
    \frac{3}{4\sqrt{\tilde{\Lambda}_2^2+\frac{3}{4}}}
    \right)
    -
    \erf\left(
    \frac{3}{4\sqrt{\tilde{\Lambda}_2^2+\frac{3}{4}}}
    \right)
    \right]
\end{align}
\begin{align}\label{eq:mean_denominator_integral}
    \Rightarrow\frac{1}{\sqrt{\pi}}\int\limits_{0}^{\tilde{\Lambda}_2R_2}{d\tilde{t}e^{-\frac{4s^2M}{N}(\tilde{t}+\tilde{\Lambda}_2)^{2}}e^{-\tilde{t}^{2}}}
    \geq
    \frac{1}{2}
    \frac{\tilde{\Lambda}_2e^{-\frac{3\tilde{\Lambda}_2^2}{4\tilde{\Lambda}_2^2+3}}}{\sqrt{\tilde{\Lambda}_2^2+\frac{3}{4}}}
    \left[
    \erf\left(R_2\sqrt{\frac{3}{4}}\right)
    -
    \erf\left(
    \frac{3}{4\sqrt{\tilde{\Lambda}_2^2+\frac{3}{4}}}
    \right)
    \right].
\end{align}

Considering the power series approximation of $\erf(x)$ for $x = \Omega(1)$ we get the following lower bound
\begin{equation}
\erf(x) \geq 1-\frac{1}{\sqrt{\pi}}\frac{e^{-x^2}}{x}.
\end{equation}
Substituting this bound in the first term of Eq. \eqref{eq:mean_denominator_integral},
\begin{align}
    \frac{1}{\sqrt{\pi}}\int\limits_{0}^{\tilde{\Lambda}_2R_2}{d\tilde{t}e^{-\frac{4s^2M}{N}(\tilde{t}+\tilde{\Lambda}_2)^{2}}e^{-\tilde{t}^{2}}}
    &\geq
    \frac{1}{2}
    \frac{\tilde{\Lambda}_2e^{-\frac{3\tilde{\Lambda}_2^2}{4\tilde{\Lambda}_2^2+3}}}{\sqrt{\tilde{\Lambda}_2^2+\frac{3}{4}}}
    \left[
    1-\frac{2}{\sqrt{3\pi}}\frac{e^{-\frac{3}{4}R_2^2}}{R_2}
    -
    \erf\left(
    \frac{3}{4\sqrt{\tilde{\Lambda}_2^2+\frac{3}{4}}}
    \right)
    \right].
\end{align}

The second term of $\mathcal{M}'_2$ in Eq. \eqref{eq:numdem} can be bounded as
\begin{align}
\begin{split}
    \zeta_N\sum\limits_{l=0}^{R_2-1}{e^{-\frac{4s^2M}{N}(\tilde{\Lambda}_2(l+1))^{2}}}
    \leq
    \zeta_N\int\limits_0^{R_2-1} dl e^{-\frac{4s^2M}{N}(\tilde{\Lambda}_2(l+1))^{2}}
    +
    \zeta_N\frac{e^{-\frac{4s^2M}{N}\tilde{\Lambda}_2^{2}}+e^{-\frac{4s^2M}{N}\tilde{\Lambda}_2^{2}R_2^2}}{2}
    +
    \\
    +
    \zeta_N\frac{2s^2M}{3N}\tilde{\Lambda}_2^2
    \left(
    e^{-\frac{4s^2M}{N}\tilde{\Lambda}_2^2}
    -
    R_2e^{-\frac{4s^2M}{N}\tilde{\Lambda}_2^2R_2^2}
    \right)
    +
    \\
    +
    4\sqrt{\pi}\zeta_Ns\sqrt{\frac{M}{N}}\tilde{\Lambda}_2
    \left[
    \erf\left(\sqrt{\frac{4s^2M}{N}}\tilde{\Lambda}_2R_2\right)
    -
    \erf\left(\sqrt{\frac{4s^2M}{N}}\tilde{\Lambda}_2\right)
    \right]
    +
    \\
    +
    \frac{8s^2M}{N}\zeta_N\tilde{\Lambda}_2^2
    \left[
    e^{-\frac{4s^2M}{N}\tilde{\Lambda}_2^2}
    -R_2e^{-\frac{4s^2M}{N}\tilde{\Lambda}_2^2R_2^2}
    \right].
\end{split}
\end{align}
\begin{align}
\begin{split}
    \zeta_N\sum\limits_{l=0}^{R_2-1}{e^{-\frac{4s^2M}{N}(\tilde{\Lambda}_2(l+1))^{2}}}
    &\leq 
    \zeta_N\sqrt{\pi}\sqrt{\frac{N}{16s^2M}}\frac{1}{\tilde{\Lambda}_2}
    \left[
    \erf\left(\sqrt{\frac{4s^2M}{N}}\tilde{\Lambda}_2R_2\right)
    -
    \erf\left(\sqrt{\frac{4s^2M}{N}}\tilde{\Lambda}_2\right)
    \right]
    +
    \\
    &+\zeta_Ne^{-\frac{4s^2M}{N}\tilde{\Lambda}_2^{2}}
    +
    \zeta_N\frac{2s^2M}{3N}\tilde{\Lambda}_2^2
    e^{-\frac{4s^2M}{N}\tilde{\Lambda}_2^2}
    \\&+
    \sqrt{\pi}\zeta_N\sqrt{\frac{16s^2M}{N}}\tilde{\Lambda}_2
    \erf\left(\sqrt{\frac{4s^2M}{N}}\tilde{\Lambda}_2R_2\right)
    +
    \frac{8s^2M}{N}\zeta_N\tilde{\Lambda}_2^2
    e^{-\frac{4s^2M}{N}\tilde{\Lambda}_2^2}.
\end{split}
\end{align}
In terms of $C_2$, the error term is
\begin{align}
\begin{split}
    \zeta_N\sum\limits_{l=0}^{R_2-1}{e^{-\frac{4s^2M}{N}(\tilde{\Lambda}_2(l+1))^{2}}}
    \leq 
    \zeta_N\sqrt{\pi}\frac{1}{2\sqrt{C_2}}
    \left[
    \erf\left(\sqrt{C_2}R_2\right)
    -
    \erf(\sqrt{C_2})
    \right]
    +
    \zeta_Ne^{-C_2}
    +
    \zeta_N\frac{C_2}{6}
    e^{-C_2}
    +
    \\
    +
    2\sqrt{\pi}\zeta_N\sqrt{C_2}
    \erf(\sqrt{C_2}R_2)
    +
    2C_2\zeta_N
    e^{-C_2}.
\end{split}
\end{align}
Substituting \ksr{with the choice} $C_2=\frac{3}{4}$,
\begin{align}
    \zeta_N\sum\limits_{l=0}^{R_2-1}{e^{-\frac{4s^2M}{N}(\tilde{\Lambda}_2(l+1))^{2}}}
    \leq 
    3.7\zeta_N
    \erf(0.86R_2)
    +
    1.3\zeta_N.
\end{align}
Combining everything, the lower bound on $\mathcal{M}'_2$ can be written as
\begin{align}
\begin{split}
    \mathcal{M}'_2
    &\geq
    \frac{1}{2}
    \frac{\tilde{\Lambda}_2e^{-\frac{3\tilde{\Lambda}_2^2}{4\tilde{\Lambda}_2^2+3}}}{\sqrt{\tilde{\Lambda}_2^2+\frac{3}{4}}}
    \left[
    1-\frac{2}{\sqrt{3\pi}}\frac{e^{-\frac{3}{4}R_2^2}}{R_2}
    -
    \erf\left(
    \frac{3}{4\sqrt{\tilde{\Lambda}_2^2+\frac{3}{4}}}
    \right)
    \right]
    -
    3.7\zeta_N
    \erf(0.86R_2)
    -
    1.3\zeta_N
    -
    4e^{-\frac{y^2}{2}}.
\end{split}
\end{align}
Replacing $\frac{\tilde{\Lambda}_2^2}{\tilde{\Lambda}_2^2+\frac{3}{4}} = \frac{\frac{N}{4s^2M}}{\frac{N}{4s^2M}+1}$,
\begin{align}
    \mathcal{M}'_2
    &\geq
    \frac{1}{2}\sqrt{\frac{\frac{N}{4s^2M}}{\frac{N}{4s^2M}+1}}
    e^{-\frac{3}{4}\frac{\frac{N}{4s^2M}}{\frac{N}{4s^2M}+1}}
    \left[
    1-\frac{2}{\sqrt{3\pi}}\frac{e^{-\frac{3}{4}R_2^2}}{R_2}
    -
    \erf\left(
    \frac{3}{4\sqrt{\tilde{\Lambda}_2^2+\frac{3}{4}}}
    \right)
    \right]
    -
    5\zeta_N
    -
    4e^{-\frac{y^2}{2}}.
\end{align}
Again assuming that $\frac{N}{4s^2M}\leq 1$, and using $\left(1+\frac{N}{4s^2M}\right)^{-\frac{1}{2}}\geq\left(1-\frac{N}{8s^2M}\right)\geq\frac{1}{2}$, and $\frac{\frac{N}{4s^2M}}{\frac{N}{4s^2M}+1}\leq\frac{N}{4s^2M}$ in the exponent
    \begin{align}
    \begin{split}
        \mathcal{M}'_2
        &\geq
        \frac{1}{4}
        \sqrt{\frac{N}{4s^2M}}
        e^{-\frac{3}{4}\frac{N}{4s^2M}}
        \left[
        1-\frac{2}{\sqrt{3\pi}}\frac{e^{-\frac{3}{4}R_2^2}}{R_2}
        -
        \erf\left(
        \sqrt{\frac{3}{4}}
        \right)
        \right]
        -
        5\zeta_N
        -
        4e^{-\frac{y^2}{2}}.
    \end{split}
    \end{align}
    Now, since $e^{-\frac{3}{4}\frac{N}{4s^2M}} \geq \left(1-\frac{3}{4}\frac{N}{4s^2M}\right)$, and $\erf(0.86R_2)\leq 1$,
    \begin{align}
    \begin{split}
        \mathcal{M}'_2
        &\geq
        \frac{1}{2}
        \sqrt{\frac{N}{4s^2M}}
        \left(1-\frac{3}{4}\frac{N}{4s^2M}\right)
        \left[
        0.11
        -
        0.33\frac{e^{-\frac{3}{4}R_2^2}}{R_2}
        \right]
        -
        5\zeta_N
        -
        4e^{-\frac{y^2}{2}}.
    \end{split}
    \end{align}
    The final bound on $\mathcal{M}'_2$ is the following
    \begin{align}\label{eq:M21}
        \mathcal{M}'_2
        \geq
        0.05\sqrt{\frac{N}{4s^2M}}
        -
        \mathcal{O}\left(\left(\frac{N}{M}\right)^{\frac{3}{2}}\right)
        -
        \mathcal{O}(\zeta_N)
        -
        \mathcal{O}(e^{-\frac{y^2}{2}})
        -
        \mathcal{O}\left(\frac{e^{-\frac{R_2^2}{2}}}{R_2}\sqrt{\frac{N}{M}}\right).
    \end{align}
Considering the bound on $y$ from Eq. \eqref{eq:y_bound}, the overall average for the case when $M=\mathcal{O}\left(\frac{N}{\zeta_N^2}  \right)$ and $M =  \Omega(N)$, is
\begin{align}
    \vert\mu\vert
    \leq
    \frac{\sqrt{\frac{N}{2s^2M}}
    \Lambda_1
    +
    \mathcal{O}(\zeta_N\Lambda_1)}
    {0.05\sqrt{\frac{N}{4s^2M}}
        -
        \mathcal{O}\left(\left(\frac{N}{M}\right)^{\frac{3}{2}}\right)
        -
        \mathcal{O}(\zeta_N)
        -
        \mathcal{O}\left(\frac{e^{-\frac{R_2^2}{2}}}{R_2}\sqrt{\frac{N}{M}}\right)},
\end{align}
\begin{align}
    \vert\mu\vert
    \leq
    20\sqrt{2}\Lambda_1
    \left[
    1
    +
    \mathcal{O}\left(\sqrt{\frac{M}{N}}\zeta_N\right)
    +
    \mathcal{O}\left(\frac{N}{M}\right)
    +
        \mathcal{O}\left(\frac{e^{-\frac{R_2^2}{2}}}{R_2}\right)
    \right],
\end{align}
\begin{align}
    \vert\mu\vert
    \leq
    30\Lambda_1
    \left[
    1
    +
    \mathcal{O}\left(\frac{N}{M}\right)
    \right]\Rightarrow\vert\mu\vert= \mathcal{O} \left( \frac{N}{\sqrt{M}} \right).
\end{align}
\subsection{Upper bound on variance}
We now find an upper bound on the variance $\delta^2$ with a similar method. First, consider
\begin{align} \label{eq:mainvarianceeqn}
    \delta^2 &= \bra{p_{M,y}}H^2\ket{p_{M,y}} - \bra{p_{M,y}}H\ket{p_{M,y}}^2
    \leq 
    \bra{p_{M,y}}H^2\ket{p_{M,y}} \\
    & \le \frac{\sum\limits_{E_{j}\le 0}{|b_{j}|^2g_y^2\left(\frac{E_j}{N}\right)(E_j)^2}+\sum\limits_{E_{j}>0}{|b_{j}|^2g_y^2\left(\frac{E_j}{N}\right)(E_j)^2}}{\sum\limits_{E_{j}\le 0}{|b_{j}|^2g_y^2\left(\frac{E_j}{N}\right)}+\sum\limits_{E_{j}>0}{|b_{j}|^2g_y^2\left(\frac{E_j}{N}\right)}}
    :=
    \frac{\mathcal{V}_{1(-)}+\mathcal{V}_{1(+)}}{\mathcal{V}_{2(-)}+\mathcal{V}_{2(+)}}.\label{eq:main_variance_fraction_form}
\end{align}

We begin by bounding the positive energy contributions, $\mathcal{V}_{1(+)}$ and $\mathcal{V}_{2(+)}$. \ksr{The corresponding negative energy contributions follow directly, because the functions inside the sum are even}.
Using the inequalities from Eq. \eqref{eq:ineq_grouped_cosine}, we can rewrite the expression for $\mathcal{V}_{1(+)}$ and $\mathcal{V}_{2(+)}$ in terms of cosine power operators
\begin{align}
   \mathcal{V}_{1(+)}
    &\leq 
    \sum\limits_{E_{j}>0}{\left\vert b_{j}\right\vert^2\cos^{2M}{\left(\frac{E_j}{N}\right)}(E_j)^2}+4e^{-y^2}\sum\limits_{E_j>0}|b_j|^2(E_j)^2
    \\
    &\leq 
    \sum\limits_{E_{j}>0}{|b_{j}|^2\cos^{2M}{\left(\frac{E_j}{N}\right)}(E_j)^2}+4e^{-y^2}\sigma^2 
    =: \mathcal{I}_1,
\end{align}
where for the second term we use that 
$
    \sum_{E_j>0}|b_j|^2(E_j)^2 
    \leq
    \sigma^2$. Similarly, the denominator can be lower bounded as
\begin{align}
    \mathcal{V}_{2(+)}
    &\leq
    \sum\limits_{E_{j}>0}{|b_{j}|^2\cos^{2M}{\left(\frac{E_j}{N}\right)}}-4e^{-\frac{y^2}{2}}\sum\limits_{E_j>0}{|b_j|^2}
    \\
    &\leq
    \sum\limits_{E_{j}>0}{|b_{j}|^2\cos^{2M}{\left(\frac{E_j}{N}\right)}}-4e^{-\frac{y^2}{2}}
    =: \mathcal{I}_2
\end{align}
where the second inequality is obtained by upper bounding the cumulative distribution of $|b_j|^2$ by 1.
First we focus on simplifying $\mathcal{I}_{1}$. We again apply the segmentation method, in a similar manner as for the average in Eq. \eqref{eq:segmentation_method}. Dividing the total energy range $[-N,N]$ into $R_1$ segments of width $\Lambda_1> 0$ each, the first term of $\mathcal{I}_1$ is then bounded as
\begin{align}
    \sum\limits_{E_{j}>0}{|b_j|^2\cos^{2M}{\left(\frac{E_j}{N}\right)}(E_j)^2} 
    &\leq 
    \sum\limits_{l=0}^{R_1-1}{\cos^{2M}{\left(\frac{\Lambda_1l}{N}\right)}(\Lambda_1(l+1))^2}\sum\limits_{\Lambda_1 l<E_j<\Lambda_1(l+1)}{|b_j|^2},
    \\
     &\leq \sum\limits_{l=0}^{R_1-1}{\cos^{2M}{\left(\frac{\Lambda_1l}{N}\right)}(\Lambda_1(l+1))^2\left[\int\limits_{\Lambda_1l}^{\Lambda_1(l+1)}{\frac{dt}{\sigma\sqrt{2\pi}}e^{-\frac{t^2}{2\sigma^{2}}}}+\zeta_N\right]}.
\end{align}
Let $\tilde{t} = \frac{t}{\sqrt{2}\sigma},  d\tilde{t} = \frac{dt}{\sqrt{2}\sigma}$, and $\tilde{\Lambda}_1 = \frac{\Lambda_1}{\sqrt{2}\sigma}$, then LHS is bounded as
\begin{align}
\begin{split}
    \sum\limits_{E_{j}>0}{|b_j|^2\cos^{2M}{\left(\frac{E_j}{N}\right)}(E_j)^2} 
   \leq \frac{1}{\sqrt{\pi}}\sum\limits_{l=0}^{R_1-1}{\int\limits_{\tilde{\Lambda}_1 l}^{\tilde{\Lambda}_1(l+1)}{d\tilde{t}\cos^{2M}{\left(\sqrt{2}\sigma\left(\frac{\tilde{t}-\tilde{\Lambda}_1}{N}\right)\right)}(\sqrt{2}\sigma(\tilde{t}+\tilde{\Lambda}_1))^2 e^{-\tilde{t}^{2}}}}
   \\
   +
   2\sigma^2\zeta_N\sum\limits_{l=0}^{R_1-1}{\cos^{2M}{\left(\frac{\sqrt{2}\sigma\tilde{\Lambda}_1l}{N}\right)}(\tilde{\Lambda}_1(l+1))^2}.
\end{split}
\end{align}
We get the following upper bound on $\mathcal{I}_1$,
\begin{align}
\begin{split}
    \mathcal{I}_1
   \leq 
   \frac{1}{\sqrt{\pi}}&\int\limits_{0}^{\tilde{\Lambda}_1 R_1}{d\tilde{t}\cos^{2M}{\left(\sqrt{2}\sigma\left(\frac{\tilde{t}-\tilde{\Lambda}_1}{N}\right)\right)}(\sqrt{2}\sigma(\tilde{t}+\tilde{\Lambda}_1))^2 e^{-\tilde{t}^{2}}}
   \\&+
   2\sigma^2\zeta_N\sum\limits_{l=0}^{R_1-1}{\cos^{2M}{\left(\frac{\sqrt{2}\sigma\tilde{\Lambda}_1l}{N}\right)}(\tilde{\Lambda}_1(l+1))^2}
   +4e^{-y^2}\sigma^2.
\end{split}
\end{align}
Using inequalities from Eq. \eqref{eq:cosine-property2} to replace cosines by exponentials
\begin{equation}
    \mathcal{I}_1
   \leq 
   \frac{1}{\sqrt{\pi}}\int\limits_{0}^{\tilde{\Lambda}_1R_1}{d\tilde{t} e^{-2M\sigma^2(\frac{\tilde{t}-\tilde{\Lambda}_1}{N})^{2}}(\sqrt{2}\sigma(\tilde{t}+\tilde{\Lambda}_1))^2 e^{-\tilde{t}^{2}}}
   +
   2\sigma^2\zeta_N\sum\limits_{l=0}^{R_1-1}{e^{-4M\sigma^2\frac{\tilde{\Lambda}_1^2l^2}{N^2}}(\tilde{\Lambda}_1(l+1))^2}+4e^{-y^2}\sigma^2.
\end{equation}

Now consider $\sigma = s \sqrt{N}$ so that 
\begin{equation}\label{eq:variance_numerator_main}
    \mathcal{I}_1
    \leq 
    \frac{2s^2N}{\sqrt{\pi}}\int\limits_{0}^{\tilde{\Lambda}_1R_1}{d\tilde{t} e^{-\frac{2s^2 M}{N}(\tilde{t}-\tilde{\Lambda}_1)^{2}}(\tilde{t}+\tilde{\Lambda}_1)^2 e^{-\tilde{t}^{2}}}
    +
    2s^2N\zeta_N\sum\limits_{l=0}^{R_1-1}{e^{-\frac{4s^2M}{N}\tilde{\Lambda}_1^2l^2}(\tilde{\Lambda}_1(l+1))^2}
    +4s^2 Ne^{-y^2}.
\end{equation}
\ksr{Using the Euler-Maclaurin formula from equation~\eqref{eq:euler_maclaurin}}, the integral in the first term is bounded as 
\begin{align}
\begin{split}
    \frac{2s^2N}{\sqrt{\pi}}\int\limits_{0}^{\tilde{\Lambda}_1R_1}{d\tilde{t} e^{-\frac{2s^2 M}{N}(\tilde{t}-\tilde{\Lambda}_1)^{2}}(\tilde{t}+\tilde{\Lambda}_1)^2 e^{-\tilde{t}^{2}}}
    \le
    \frac{2s^2N}{\sqrt{\pi}}\frac{e^{-\frac{2s^2M}{N}\tilde{\Lambda}_1^2}}{4(1+\frac{2s^2M}{N})^{\frac{5}{2}}}
    \Bigg[
    4\tilde{\Lambda}_1\sqrt{1+\frac{2s^2M}{N}}\left(1+\frac{3s^2M}{N}\right)
    +
    \\
    -
    2e^{-\tilde{\Lambda}_1^2R_1^2(1+\frac{2s^2M}{N})}e^{\tilde{\Lambda}_1^2R_1\frac{4s^2M}{N}}
    \tilde{\Lambda}_1\sqrt{1+\frac{2s^2M}{N}}\left(R_1\left(1+\frac{2s^2M}{N}\right)+2\left(1+\frac{3s^2M}{N}\right)\right)
    +
    \\
    +
    e^{\frac{1}{1+\frac{2s^2M}{N}}(\frac{2s^2M}{N}\tilde{\Lambda}_1)^2}
    \sqrt{\pi}
    \left(
    \left(
    1+\frac{2s^2M}{N}
    \right)
    +
    2\tilde{\Lambda}_1^2
    \left(
    1+\frac{4s^2M}{N}
    \right)^2
    \right)
    \Bigg[    
    \erf\left(
    \frac{\frac{2s^2M}{N}\tilde{\Lambda}_1}{\sqrt{1+\frac{2s^2M}{N}}}
    \right)
    +
    \\
    +
    \erf\left(
    \frac{\tilde{\Lambda}_1R_1(1+\frac{2s^2M}{N})-\frac{2s^2M}{N}\tilde{\Lambda}_1}{\sqrt{1+\frac{2s^2M}{N}}}
    \right)
    \Bigg]
    \Bigg],
\end{split}
\end{align}
\begin{align}
\begin{split}
    \frac{2s^2N}{\sqrt{\pi}}\int\limits_{0}^{\tilde{\Lambda}_1R_1}{d\tilde{t} e^{-\frac{2s^2 M}{N}(\tilde{t}-\tilde{\Lambda}_1)^{2}}(\tilde{t}+\tilde{\Lambda}_1)^2 e^{-\tilde{t}^{2}}}
    &\leq
    \frac{2s^2N}{\sqrt{\pi}}\frac{e^{-\frac{2s^2M}{N}\tilde{\Lambda}_1^2}}{4(1+\frac{2s^2M}{N})^{\frac{5}{2}}}
    \left[
    4\tilde{\Lambda}_1\sqrt{1+\frac{2s^2M}{N}}\left(1+\frac{3s^2M}{N}\right)
    \right.
    \\
    &\quad
    -
    2e^{-\tilde{\Lambda}_1^2R_1^2(1+\frac{2s^2M}{N})}e^{\tilde{\Lambda}_1^2R_1\frac{4s^2M}{N}}
    \tilde{\Lambda}_1R_1\left(1+\frac{2s^2M}{N}\right)^{\frac{3}{2}}
    \\
    &\quad
    \left.
    +
    2e^{\frac{1}{1+\frac{2s^2M}{N}}(\frac{2s^2M}{N}\tilde{\Lambda}_1)^2}
    \sqrt{\pi}
    \left(
    \left(
    1+\frac{2s^2M}{N}
    \right)
    \right.
    \right.
    \\&\quad\left.\left.+
    2\tilde{\Lambda}_1^2
    \left(
    1+\frac{4s^2M}{N}
    \right)^2
    \right)
    \right].
\end{split}
\end{align}

Substituting $C_1 = \frac{2s^2M}{N}\tilde{\Lambda}_1^2$ 
\begin{align}\label{eq:main_term_variance_num}
    \frac{2s^2N}{\sqrt{\pi}}\int\limits_{0}^{\tilde{\Lambda}_1R_1}{d\tilde{t} e^{-\frac{2s^2 M}{N}(\tilde{t}-\tilde{\Lambda}_1)^{2}}(\tilde{t}+\tilde{\Lambda}_1)^2 e^{-\tilde{t}^{2}}}
    \leq
    \frac{3s^2N}{\sqrt{\pi}}\frac{\tilde{\Lambda}_1^3e^{-C_1}}{\tilde{\Lambda}_1^2+C_1}
    +
    s^2N\frac{\tilde{\Lambda}_1^3e^{-\frac{\tilde{\Lambda}_1^2C_1}{\tilde{\Lambda}_1^2+C_1}}}{(\tilde{\Lambda}_1^2+C_1)^{\frac{3}{2}}}
    +
    8s^2N\frac{\tilde{\Lambda}_1^3e^{-\frac{\tilde{\Lambda}_1^2C_1}{\tilde{\Lambda}_1^2+C_1}}}{\sqrt{\tilde{\Lambda}_1^2+C_1}}
    .
\end{align}
The error term of Eq. \eqref{eq:variance_numerator_main} is bounded as
\begin{align}\label{eq:err_term_variance_num}
\begin{split}
   2s^2N\zeta_N\sum\limits_{l=0}^{R_1-1}{e^{-\frac{4s^2M}{N}\tilde{\Lambda}_1^2l^2}(\tilde{\Lambda}_1(l+1))^2}
    &\leq
    \frac{s^2(8\sqrt{C_1}+\sqrt{2\pi}(1+4C_1))}{8C_1\sqrt{C_1}}N\zeta_N\tilde{\Lambda}_1^2
    \\&+
    s^2N\zeta_N\tilde{\Lambda}_1^2\left(
    e^{-2C_1(R_1-1)^2}R_1^2
    +
    1
    \right)
    \\
    &+
    \frac{s^2N\zeta_N}{6}
    (2e^{-\frac{4s^2M}{N}\tilde{\Lambda}_1^2(R_1-1)^2}\tilde{\Lambda}_1^2R_1)
    \\
    &+
    2s^2N\zeta_N\tilde{\Lambda}_1^2\frac{\zeta(2)}{\pi^2}\left[5+\sqrt{\frac{\pi}{8C_1}}(5+4C_1)\right].
\end{split}
\end{align}
\ksr{Combining Eq.~\eqref{eq:main_term_variance_num} and Eq.~\eqref{eq:err_term_variance_num},} we get the following bound on $\mathcal{I}_1$,
\begin{align}
\begin{split}
    \mathcal{I}_1
    &\leq
    \frac{3s^2N}{\sqrt{\pi}}\frac{\tilde{\Lambda}_1^3e^{-C_1}}{\tilde{\Lambda}_1^2+C_1}
    +
    s^2N\frac{\tilde{\Lambda}_1^3e^{-\frac{\tilde{\Lambda}_1^2C_1}{\tilde{\Lambda}_1^2+C_1}}}{(\tilde{\Lambda}_1^2+C_1)^{\frac{3}{2}}}
    +
    8s^2N\frac{\tilde{\Lambda}_1^3e^{-\frac{\tilde{\Lambda}_1^2C_1}{\tilde{\Lambda}_1^2+C_1}}}{\sqrt{\tilde{\Lambda}_1^2+C_1}}
    +
    \\
    &+
    \frac{s^2(8\sqrt{C_1}+\sqrt{2\pi}(1+4C_1))}{8C_1\sqrt{C_1}}N\zeta_N\tilde{\Lambda}_1^2
    +
    s^2N\zeta_N\tilde{\Lambda}_1^2\left(
    e^{-2C_1(R_1-1)^2}R_1^2
    +
    1
    \right)
    +
    \\
    &+
    \frac{s^2N\zeta_N}{6}
    (2e^{-\frac{4s^2M}{N}\tilde{\Lambda}_1^2(R_1-1)^2}\tilde{\Lambda}_1^2R_1)
    +
    2s^2N\zeta_N\tilde{\Lambda}_1^2\frac{\zeta(2)}{\pi^2}\left[5+\sqrt{\frac{\pi}{8C_1}}(5+4C_1)\right]
    +
    4s^2 Ne^{-y^2}.
\end{split}
\end{align}
To simplify the above form note that
\begin{equation}
    e^{-\frac{\tilde{\Lambda}_1^2C_1}{\tilde{\Lambda}_1^2+C_1}}
    \leq
    1 \quad \quad \text{and}  \quad \quad  \frac{\tilde{\Lambda}_1}{\sqrt{\tilde{\Lambda}_1^2+C_1}}
    =
    \sqrt{\frac{1}{1+\frac{2s^2M}{N}}}
    \leq
    1,
\end{equation}
for any scaling of $M$.
Additionally, for $\frac{N}{2s^2M}<1$ there is the tighter bound
\begin{align}
    \frac{\tilde{\Lambda}_1}{\sqrt{\tilde{\Lambda}_1^2+C_1}}
    =
    \sqrt{\frac{1}{1+\frac{2s^2M}{N}}}
    \leq
    \sqrt{\frac{N}{2s^2M}}.
\end{align}
So, for $\frac{N}{2s^2M}<1$ the numerator is bounded as
\begin{align}
\begin{split}
    \mathcal{I}_1
    &\leq
    \frac{3s^2N}{\sqrt{\pi}}\frac{\tilde{\Lambda}_1^3}{\tilde{\Lambda}_1^2+C_1}
    +
    s^2N\frac{\tilde{\Lambda}_1^3}{(\tilde{\Lambda}_1^2+C_1)^{\frac{3}{2}}}
    +
    8s^2N\frac{\tilde{\Lambda}_1^3}{\sqrt{\tilde{\Lambda}_1^2+C_1}}
    +
    \\
    &+
    \frac{s^2(8\sqrt{C_1}+\sqrt{2\pi}(1+4C_1))}{8C_1\sqrt{C_1}}N\zeta_N\tilde{\Lambda}_1^2
    +
    s^2N\zeta_N\tilde{\Lambda}_1^2\left(
    e^{-2C_1(R_1-1)^2}R_1^2
    +
    1
    \right)
    +
    \\
    &+
    \frac{s^2N\zeta_N}{6}
    (2e^{-\frac{4s^2M}{N}\tilde{\Lambda}_1^2(R_1-1)^2}\tilde{\Lambda}_1^2R_1)
    +
    2s^2N\zeta_N\tilde{\Lambda}_1^2\frac{\zeta(2)}{\pi^2}\left[5+\sqrt{\frac{\pi}{8C_1}}(5+4C_1)\right]
    +
    4s^2 Ne^{-y^2}.
\end{split}
\end{align}
Now applying the following inequalities to the first three terms
\begin{align}
    \frac{\tilde{\Lambda}_1^3}{\tilde{\Lambda}_1^2+C_1}
    \leq
    \frac{N}{2s^2M}\tilde{\Lambda}_1,
    \\
    \frac{\tilde{\Lambda}_1^3}{(\tilde{\Lambda}_1^2+C_1)^{\frac{3}{2}}}
    \leq
    \left(\frac{N}{2s^2M}\right)^{\frac{3}{2}},
    \\
    \frac{\tilde{\Lambda}_1^3}{\sqrt{\tilde{\Lambda}_1^2+C_1}}
    \leq
    \tilde{\Lambda}_1^2
    \sqrt{\frac{N}{2s^2M}}.
\end{align}
Consider the following,

(i) $C_1 = \mathcal{O}(1)$ which means that $ \tilde{\Lambda}_1=\mathcal{O}(1)$ for $M=\Omega(N)$, 

(ii) $y=\Omega\left(\sqrt{\log\frac{M^3}{N^2}}\right)$ from Eq. \eqref{eq:y_bound},

(iii) $\tilde{\Lambda}_1R_1 \propto \sqrt{N}$, then
\begin{align}
\begin{split}
    \mathcal{I}_1
    \leq
    \frac{3}{2\sqrt{\pi}}\tilde{\Lambda}_1
    \frac{N^2}{M}
    +
    \frac{1}{2\sqrt{2}s}\frac{N^{\frac{5}{2}}}{M^{\frac{3}{2}}}
    +
    4\sqrt{2}s\tilde{\Lambda}_1^2\frac{N^{\frac{3}{2}}}{M^{\frac{1}{2}}}
    +
    \mathcal{O}(N\zeta_N\tilde{\Lambda}_1^2)
    +
    \mathcal{O}(N^2\zeta_Ne^{-2C_1(R_1-1)^2}).
\end{split}
\end{align}
Substituting $\tilde{\Lambda}_1 = \sqrt{C_1}\sqrt{\frac{N}{2s^2M}}$,
\begin{align}\label{eq:final_numerator_variance}
\begin{split}
    \mathcal{I}_1
    \leq
    \left(\frac{3\sqrt{C_1}}{2s\sqrt{2\pi}}
    +
    \frac{1}{2\sqrt{2}s}
    +
    \frac{2\sqrt{2}C_1}{s}\right)\frac{N^{\frac{5}{2}}}{M^{\frac{3}{2}}}
    +
    \mathcal{O}\left(\frac{N^2}{M}\zeta_N\right)
    +
    \mathcal{O}(N^2\zeta_Ne^{-2C_1(R_1-1)^2}).
\end{split}
\end{align}

We have obtained an upper bound on $\mathcal{I}_1$ which upper bounds $\mathcal{V}_{1(+)}$ and $\mathcal{V}_{1(-)}$ \ksr{(numerator of Eq.~\eqref{eq:main_variance_fraction_form})}.
For the denominator, note that $\mathcal{M}'_2$ in Eq. \eqref{eq:M21} from the average calculation is also a lower bound on the denominator terms $\mathcal{V}_{2(+)}$ and $\mathcal{V}_{2(-)}$ of the variance. Combining all the expressions, we can obtain the upper bound on the variance for $M=\mathcal{O}\left(\frac{N}{\zeta_N^2}  \right)$ and $M =  \Omega(N)$ as follows,
\begin{equation}\label{eq:final_variance_unsimplified}
        \delta^2
        \leq
        \frac{
        \begin{split}
    \left(\frac{3\sqrt{C_1}}{2s\sqrt{2\pi}}+\frac{1}{2\sqrt{2}s}+\frac{2\sqrt{2}C_1}{s}\right)\frac{N^{\frac{5}{2}}}{M^{\frac{3}{2}}}
        \left[
        1
        +
        \mathcal{O}\left(\sqrt{\frac{M}{N}}\zeta_N\right)
        +
        \mathcal{O}\left(\frac{M^{\frac{3}{2}}}{N^{\frac{1}{2}}}\zeta_Ne^{-2C_1(R_1-1)^2}\right)\right.
        \\\left.+
        \mathcal{O}\left(\frac{M^{\frac{3}{2}}}{N^{\frac{3}{2}}}\zeta_NR_1e^{-\frac{2s^2M}{N}R_1}\right)
        \right]
        \end{split}
        }
        {0.1\sqrt{\frac{N}{4s^2M}}
        \left[
        1
        -
        \mathcal{O}\left(\frac{N}{M}\right)
        -
        \mathcal{O}\left(\zeta_N\sqrt{\frac{M}{N}}\right)
        -
        \mathcal{O}\left(\frac{e^{-\frac{R_2^2}{2}}}{R_2}\right)
        \right]},
    \end{equation}
    \begin{align}
        \Rightarrow
        \delta^2
        \leq
        \eta\frac{N^2}{M}
        \left[
        1
        +
        \mathcal{O}\left(\frac{N}{M}\right)
        \right],
    \end{align}
where we define the constant
    \begin{equation}
        \eta := \frac{15\sqrt{C_1}}{\sqrt{2\pi}}+\frac{5}{\sqrt{2}}+20\sqrt{2}C_1.
    \end{equation}
    
\end{document}